\documentclass[12pt]{article}
\usepackage[a4paper]{geometry}
\usepackage{mathtools}
\usepackage[english]{babel}

\usepackage{comment}
\usepackage{graphicx, amsfonts, amsmath, amssymb, amsthm, color,bm}
\usepackage[colorlinks=true, allcolors=blue]{hyperref}

\usepackage{pifont}% http://ctan.org/pkg/pifont
\newcommand{\cmark}{$\ding{51}$}%
\newcommand{\xmark}{$\ding{55}$}%

\usepackage[ruled]{algorithm2e}
\SetAlgoSkip{bigskip}

\usepackage{tikz}
\usetikzlibrary{arrows.meta}
\usepackage{nicematrix}

\newtheorem{theorem}{Theorem}[section]

\newtheorem{lemma}[theorem]{Lemma}

\theoremstyle{definition}
\newtheorem{example}[theorem]{Example}
\newtheorem{definition}[theorem]{Definition}

\theoremstyle{remark}
\newtheorem{remark}[theorem]{Remark}

\newcommand{\pr}{{\mathbb P}}
\newcommand{\ep}{{\mathbb E}}
\DeclareMathOperator{\bin}{Bin}
\newcommand{\KK}{{\mathcal K}}
\newcommand{\Dset}{{\mathcal D}}
\newcommand{\vc}[1]{{\boldsymbol #1}}
\newcommand{\ol}[1]{\overline{ #1}}
\newcommand{\Perr}{\pr({\rm err})}
\newcommand{\PERR}[2]{\pr({\rm err}; \mbox{{\rm #1}}, #2)}
\newcommand{\Psuc}{\pr({\rm suc})}

\DeclareMathOperator{\PD}{PD}
\newcommand{\wh}{\widehat}
\newcommand{\Tmag}[1]{T^*_{\rm #1}}
\newcommand{\defvec}{{\mathcal U}}
\newcommand{\AAA}[1]{{\mathcal A}_{\vc{#1}}}
\newcommand{\AAS}[1]{|\AAA{#1}|}
\newcommand{\suc}{{\rm suc}}
\newcommand{\II}{{\mathbb I}}
\newcommand{\pc}{\psi}
\newcommand{\new}{\nu}  % p = \new / K
\newcommand{\aequi}{\simeq}  % or \sim ?

\newcommand{\Tinf}{T_\infty}
\newcommand{\Tfin}{T_{\mathrm{fin}}}

\newcommand{\ee}{\mathrm{e}}

\title{Small error algorithms for tropical group testing}
\author{Vivekanand Paligadu\thanks{School of Mathematics, University of Bristol, Bristol, UK.} \and Oliver Johnson\thanks{Corresponding author: School of Mathematics, University of Bristol, Fry Building, Woodland Road, Bristol, BS8 1UG, UK.
Email: {\tt O.Johnson@bristol.ac.uk}} \and Matthew Aldridge\thanks{School of Mathematics, University of Leeds, Leeds, UK}
}

\begin{document}
\maketitle

\begin{abstract}
We consider a version of the classical group testing problem motivated by PCR testing for COVID-19. In the so-called tropical group testing model, the outcome of a test is the lowest cycle threshold (Ct) level of the individuals pooled within it, rather than a simple binary indicator variable. We introduce the tropical counterparts of three classical non-adaptive algorithms (COMP, DD and SCOMP), and analyse their behaviour through both simulations and bounds on error probabilities. By comparing the results of the tropical and classical algorithms, we gain insight into the extra information provided by learning the outcomes (Ct levels) of the tests. We show that in a limiting regime the tropical COMP algorithm requires as many tests as its classical counterpart, but that for sufficiently dense problems tropical DD can recover more information with fewer tests, and can be viewed as essentially optimal in certain regimes.
\end{abstract}

\section{Introduction}

\subsection{Group testing and tropical group testing}

Group testing is the problem of reliably recovering a subset $\KK$ of `defective' items from a population of $N$ items, using a relatively small number $T$ of so-called pooled tests to test multiple items at the same time. In the classical noiseless setting, the outcome of such tests is binary, indicating whether or not there is at least one defective item in the pool.

Group testing was initially introduced by Dorfman \cite{dorfman} in the context of testing for syphilis. It has now developed into a combinatorial and algorithmic problem with a rich history \cite{du,aldridge2019}, where one seeks to understand the trade-offs between $N$, $|\KK|$ and $T$, and to understand how $\KK$ can be efficiently recovered using computationally feasible algorithms. 

Group testing has applications to many fields such as biology, manufacturing, communications and information technology, as described in \cite[Section 1.7]{aldridge2019}. This framework also recently gained considerable attention during the COVID-19 pandemic -- see \cite{aldridge2022} for a review of its use in this context, and \cite{yelin2020} for an early proof of concept that pooled testing can detect the presence of COVID-positive individuals. In general, the efficiency of group testing makes it useful when using PCR machines to test large populations for rare conditions.

% MPA CHANGE 1 - Next six paragraphs
In this paper, however, we consider not the standard binary group testing model, but rather the `tropical group testing' model of Wang \emph{et al.}\ \cite{wang}. While the binary model considers only `positive or negative' outcomes, the tropical group testing model considers different levels of strength of infection (or `seriousness of defectiveness'). In the tropical model, the output is given by the strongest infected (or `most defective') input.

We have three reasons for wanting to study further the tropical model of Wang \emph{et al.}\ \cite{wang}. First, the tropical model is a fascinating mathematical puzzle in its own right, where the generalisation of the binary group testing model creates interesting problems requiring creative solutions. We are pleased to be able to build on the work of \cite{wang} to propose and rigorously analyse new algorithms for the tropical model.

Second, tropical group testing is an excellent model for situations where the performance of a set of items is determined by the most extreme performance of a single item in the set. Consider, for example, the capacity of a route through a network, which is typically limited by the lowest bandwidth link on that route. If $u_i$ is the bandwidth of link $i$, and $x_{ti} = 1$ denotes that link $i$ is on route $t$, then the capacity of route $t$ would be given by $y_t = \min\{u_i : x_{ti} = 1\}$. This is precisely the tropical group testing model defined below in Definition \ref{def:outcome}.

%As well as approximating an additive model like \eqref{eq:bhar2}, 

%Classical group testing, however, only considers tests with binary `positive or negative' outcomes, so it can fail to take advantage of all the available information, for example regarding strength of infection.

Third, tropical group testing models how information regarding strengths of infections is combined in PCR testing, for example when testing for COVID. On one side, traditional group testing only considers binary positive-or-negative outcomes, so it does not represent how samples with different strengths of infection combine. On the other side, more complicated models (like that we describe in Subsection \ref{subsec:approx} below) do not seem amenable to theoretical analysis, only to exploration via simulation. In Subsection \ref{subsec:approx} we explain how tropical group testing can be formally seen as an approximation of a more detailed, but less tractable, model of PCR testing. In this sense, we feel that the tropical model of \cite{wang} hits the `sweet spot' between keeping the analytical tractability of binary group testing, while still adopting a model that reflects (even if it does not perfectly capture) how PCR testing works.

Other COVID-inspired pooled testing schemes that are designed to take account of quantitative information through numerical values of test outcomes include Tapestry \cite{ghosh2020,ghosh2021}, the work of Bharadwaja and Murthy \cite{bharadwaja}, and the two-stage adaptive scheme of Heidarzadeh and Narayanan \cite{heidarzadeh}.

\subsection{The tropical model as an approximation of PCR testing} \label{subsec:approx}

To motivate the tropical model, we will consider a detailed model the PCR test, and show how tropical group testing can be viewed as an approximation of this model. Readers already convinced by the utility of the tropical model who want to read about our new contributions can skip to Subsection \ref{subsec:contributions}.

In the PCR test (see, for example, \cite{white1989polymerase}), the viral genetic material is extracted from a sample and amplified in `cycles' of a process using polymerase enzymes. In each cycle, the quantity of this material is multiplied by some factor (usually doubled, or slightly less). The presence of the viral genetic material is detected by fluorescence, indicating a positive result if the quantity present exceeds a certain amount. The \emph{cycle threshold} (Ct) value is the number of cycles after which fluorescence is observed -- this represents the number of amplifications (or doublings) required to achieve detection. Hence when using PCR to test for COVID \cite{rao2020narrative}, a lower Ct value indicates a higher concentration of viral genetic material in the sample. 

The following mathematical description of the PCR protocol follows closely that by Bharadwaja and Murthy \cite{bharadwaja}. Suppose a single sample has initial viral concentration $z$. In each cycle it will be amplified by a factor of $1 + q$, say. (Bharadwaja and Murthy \cite{bharadwaja} suggest values of $q \in [1/2,1]$ are typical.) So the Ct number, the number of cycles $u$ required for the amplified viral concentration to cross a given threshold $\tau$, will satisfy $(1+q)^u z = \tau$ (see \cite[eq. (1)]{bharadwaja}).  Equivalently, $z = \tau(1+q)^{-u}$. 

This extends to a model where multiple items with different viral concentrations $z_i = \tau(1+q)^{-u_i}$ are tested according to a binary pooling matrix $\vc{x} = (x_{ti})$, where $x_{ti} = 1$ means item $i$ is in pool $t$. (See Section \ref{sec:notation} for a formal description of the test matrix.) The total initial viral concentration for pool $t$ will be
\[ \sum_{i : x_{ti} = 1} z_i = \tau \!\sum_{i : x_{ti} = 1} (1+q)^{-u_i} . \]
This will cross the threshold $\tau$ after $y_t$ cycles, where
\[ (1+q)^{y_t} \, \tau\! \sum_{i : x_{ti} = 1} (1+q)^{-u_i} = \tau . \]
Rearranged to make $y_t$ the subject, we have
\begin{equation}  y_t = -\log_{1+q} \bigg(\sum_{i: x_{ti} = 1} (1+q)^{-u_i}\bigg) . \label{eq:bhar2}
\end{equation}
This is the model of \cite[eq. (4)]{bharadwaja}, except for that also including an additive noise term $\epsilon_t \sim N(0,\sigma^2)$.

It is possible to directly study the model \eqref{eq:bhar2} empirically. For example, Bharadwaja and Murthy \cite{bharadwaja} describe several algorithms based on gradient descent, for both known and unknown $q$, and validate them empirically for prevalence levels found during the Covid-19 pandemic.  

However, an alternative perspective comes through the way that Ct values are explicitly combined together in the so-called \emph{tropical group testing} model of Wang \emph{et al.}\ \cite{wang}, which takes the outcome of each test to be equal to the minimum value of $u_i$ among items appearing in it (see Definition \ref{def:outcome} below). This construction can be motivated using the `tropicality approximation' to \eqref{eq:bhar2}:
\begin{equation}  y_t = -\log_{1+q} \bigg(\sum_{i: x_{ti} = 1} (1+q)^{-u_i}\bigg)  \approx \min_{i:x_{ti} = 1} u_i. \label{eq:trop}
\end{equation}
This approximation appears to render the model more amenable to mathematical analysis. It seems to be considerably harder to formally prove performance guarantees in the setting of \eqref{eq:bhar2} and \cite{bharadwaja} than the tropical model of \cite{wang} and this paper.

% MPA change 3 - Reviewer's complaint number 5; added "when it comes to modelling PCR tests"
The tropicality approximation \eqref{eq:trop} is exact when a pool contains $0$ or $1$ infected items (those with $u_i < \infty$), and is very accurate when the viral concentrations of $2$ or more infected items differ by a reasonable amount, but can be slightly inaccurate if $2$ or more items have very similar viral concentrations. We believe that the tropical model strikes a good balance between maintaining the analytical tractability of standard binary group testing while still modelling the effect of differing viral concentrations as in \eqref{eq:bhar2}. In that sense, the contribution of our paper is to prove rigorous bounds in a model which, when it comes to modelling PCR tests, is more realistic than the standard group testing model studied in the vast majority of the group testing literature (see \cite{aldridge2019} and elsewhere).

Since the outcome of a tropical group test is  the minimum Ct value of the infected individuals contained within it, strongly infected items with low Ct values tend to dominate the test outcomes, potentially concealing weakly infected individuals with high Ct values. In an attempt to address this limitation, Wang \emph{et al.}\ introduce the concept of a `tropical code' involving `delays' (controlled by a `delay matrix'), which correspond to adding some samples partway through a PCR run, so they get amplified fewer times. With this approach, \cite{wang} describes adaptive and non-adaptive constructions.  

\subsection{Our contributions} \label{subsec:contributions}

The key contribution of this paper is the development and analysis of non-adaptive algorithms in the tropical setting to recover the Ct values of defective items under a small-error criterion, and to demonstrate gains in performance relative to the classical group testing setting. These algorithms are tropical generalisations of the classical COMP, DD and SCOMP algorithms \cite{aldridge2014,aldridge2019}.

% MPA CHANGE 2
A major strength of our algorithms and results is that they do not require the use of a delay matrix, meaning that all samples are added to the machine simultaneously at the start of the process. This means the tests can be run in parallel simultaneously without needing to pause the machine to load further samples, making the resulting schemes easy to implement in practice on a PCR machine. (Strictly speaking, one could say adding a sample at the beginning of the process is `delay $0$' and not adding a sample at all is `delay $\infty$'; so one could say our methods involve no \emph{nontrivial} delays.)

In particular, we identify a sharp threshold for the performance of the tropical DD algorithm (see Section \ref{sec:tropdd}) in certain parameter regimes. In Theorem \ref{thm:DDachievability} we give an achievability result by showing  that in a limiting regime where the number of tests 
\begin{equation} T > \max\{\Tinf, T_d, T_{d-1}, \dots, T_1 \} \label{eq:headline} \end{equation}
then the error probability of this algorithm tends to zero. Here the $T_r$ are explicit expressions in terms of the total number of items with particular defectivity levels. Roughly speaking, $\Tinf$ tests are required to find most (but not necessarily all) of the non-defective items, while $T_r$ tests are required to find all the defective items with Ct value $r$. Further in Remark \ref{rem:improve}, we argue that in a certain `uniform' case, this result represents an explicit (albeit second-order) improvement over the performance of classical DD.

In contrast, in Theorems \ref{thm:DDconverse} and \ref{thm:genconv} we show that in the regime where
\[ T < \max\{T_d, T_{d-1}, \dots, T_1 \} \]
then the error probability of tropical DD and even of an optimal algorithm tends to $1$. Since apart from the absence of $T_\infty$ this is identical to the expression \eqref{eq:headline}, we can conclude that our tropical DD algorithm is asymptotically optimal in parameter regimes where $T_\infty$ does not give the maximum in \eqref{eq:headline}.

The structure of the rest of the paper is as follows. In Section \ref{sec:notation} we introduce the notation used in the paper, and formalise the tropical group testing model. In Section \ref{sec:algo} we describe the three tropical algorithms we will study, and briefly mention some of their basic properties. Section \ref{sec:simulation} gives simulation results indicating the performance of these algorithms. We analyse the theoretical performance of the tropical COMP algorithm in Section \ref{sec:analysis_COMP} and of the tropical DD algorithm in Sections \ref{sec:analysis_DD} and \ref{sec:DDconverse}.

\section{Notation and tropical model} \label{sec:notation}

We adapt the classical group testing notation and algorithms of \cite{aldridge2014,aldridge2019} to the tropical group testing model of \cite{wang}. The tropical model replaces the `OR' operation of standard group testing with a `$\min$' operation, in a way motivated by the use of PCR testing for COVID-19.

In more detail, for a fixed positive integer value $d$ we define the set $\Dset = \{ 1, 2, \ldots, d, \infty\}$ of possible defectivity (or infection) levels. Here level $\infty$ represents the state of being not defective, and levels $1, 2, \ldots, d$ represent different levels of defectivity. As with Ct values in PCR testing, the lower the numerical value of the defectivity level, the stronger the infection; and the higher the numerical value of the defectivity level, the weaker the infection. The exact values represented in the set do not matter from a mathematical point of view -- while it may be closer to medical practice to use Ct values such as $\{ 20, 21,\ldots, 40, \infty \}$, the choice $\{1, 2, \dots, d, \infty\}$ provides notational convenience.

Given $N$ items, we represent the defectivity level $U_i \in \Dset$ of each item $i$ as a vector $\vc{U} = (U_1, \ldots, U_N) \in \Dset^N$. We write $\KK_r = \{ j: U_j = r \}$ for the set of items at each level $r \in \Dset$, and write $\KK = \bigcup_{r=1}^d \KK_r$ for the total set of defective items, with finite $U_i$. We write $K_r = | \KK_r |$ for the size of each set, $K = \sum_{r=1}^d K_r = | \KK|$ for the total number of defective items, and adopt the notation $\vc{K} = (K_1, \ldots, K_d)$. For $1 \leq r \leq d$ and $1 \leq s \leq K_r$, we will write $i(r,s)$ for the $s$th item in set $\KK_r$ (labelled arbitrarily within $\KK_r$).

We assume a combinatorial model: that is, we fix set sizes $K_1, \ldots, K_d$ in advance and assume that the sets $\KK_r$ are disjoint and chosen uniformly at random among sets which satisfy these constraints. We will sometimes consider a limiting sequence of problems where $N \rightarrow \infty$ with $K \aequi N^\alpha$ for some fixed $\alpha \in (0,1)$ and $K_i \aequi \theta_i K$ for some $\theta_i$ with $\sum_{i=1}^d \theta_i = 1$.

We use a non-adaptive testing strategy, where we fix the whole test design in advance. We represent the test design in a binary $T \times N$ test matrix $\vc{x}$, with the standard convention that $x_{ti} = 1$ means that item $i$ appears in test $t$ and $x_{ti}=0$ means that item $i$ does not appear in test $t$. Our use of non-adaptive strategies in this context is motivated by the fact  that PCR tests can be performed in parallel using plates with large numbers of wells (such as $96$ or $384$) -- see for example \cite{erlich2} -- meaning that the test strategy needs to be designed in advance.

We now describe the outcome of a so-called tropical group test.

\begin{definition} \label{def:outcome}
\emph{Tropical group testing} is defined by the outcome $Y_t$ of test $t$ being given by the lowest defectivity level $U_i$ among items $i$ that appear in the test:
\begin{equation} \label{eq:outcome}
Y_t = \min_i \{ U_i: x_{ti} = 1 \}.
\end{equation}
\end{definition}

For $d=1$, there are only two defectivity levels possible for an item $i$, namely $U_i = 1$ (defective) and $U_i = \infty$ (non-defective). In this case, Definition \ref{def:outcome} reduces to Dorfman's standard binary group testing model \cite{dorfman}, with the outcome of a negative test $t$ denoted by $Y_t = \infty$ (rather than the usual $Y_t = 0$). We refer to this as `classical group testing'.

For any value of $d$, if a test contains no defective items (that is, if $U_i = \infty$ for all items $i$ in the test) then the outcome is $Y_t = \infty$, which we regard as a negative test, just as in classical group testing. However, unlike classical group testing, we also receive information about the defectivity levels of the items through the outcomes of positive tests being a number from $1$ to $d$.

In order to analyse tropical group testing, we make some definitions that will be useful, and which extend the definitions and terminology of \cite{aldridge2014}.

\begin{definition} \label{def:mudef}
Write $\mu_i$ for the highest outcome of any test that item $i$ appears in:
\begin{equation} \label{eq:mudef}
\mu_i \coloneqq \max_t \{ Y_t : x_{ti} = 1 \}. \end{equation} 
If item $i$ is not tested, so that $\{ Y_t: x_{ti} = 1 \} = \emptyset$, we use the convention $\mu_i \coloneqq 1$.
\end{definition}

A key deduction is that $\mu_i$ is the lowest possible defectivity level for item $i$.

\begin{lemma} \label{lem:deduct}
For each item $i$, we have $U_i \geq \mu_i$.

In particular, if $\mu_i = \infty$ (that is, if the item appears in a negative test) then we can recover with certainty that $U_i = \infty$.
\end{lemma}

\begin{proof}
If an item $i$ is not tested at all, then by Definition \ref{def:mudef} we know that $\mu_i = 1$, and so the result trivially holds.

Otherwise, if an item $i$ is tested, then for each $t$ such that $x_{ti} = 1$, by Definition \ref{def:outcome}, we know that $U_i \geq Y_t$. So $U_i \geq 
\max_t \{ Y_t: x_{ti} = 1 \} = \mu_i$. 
\end{proof}

\begin{definition} \label{def:keydefs} We define the following:
\begin{enumerate}
\item \label{it:dndr} For each $1 \leq r \leq d$, we refer to an item $i$ that has $\mu_i = r$ as $\PD(r)$ (`Possibly Defective at levels $\{r, \ldots, d, \infty \}$').% and an item $i$ with $\mu_i > r$ as $\DND(r)$ (`Definitely Not Defective at levels $\{1, \ldots, r\}$').
\item \label{it:intmas} For $r \in \Dset$, we say that an item $i$ of defectivity level $U_i = r$ is {\em intruding} if it never appears in a test of outcome $r$ (in which case strict inequality $U_i > \mu_i$ holds in Lemma \ref{lem:deduct}).
\item \label{it:hdef}
For $r \in \Dset$, write $H_r$ for the number of tested non-defective items in $\PD(r)$ (those that have $\mu_i = r$). For convenience, also define $H_0$ to be the number of untested non-defective items.
\end{enumerate}
\end{definition}

 If $d=1$, then, in the notation of \cite{aldridge2014}, the number $G$ of intruding non-defectives (i.e.\ non-defectives that don't appear in any negative tests) corresponds here to those items $i$ with $\mu_i=1$, tested or untested; so $G$ in \cite{aldridge2014} corresponds to $H_0 + H_1$ here.

To aid understanding, it can be helpful to sort the rows and columns of the test matrix as illustrated in Figure \ref{fig:sorted}. 

\begin{figure}
\begin{center}
\begin{comment}
$$ \left( \begin{array}{cccc|cccc} 
K_1 & K_2 & \ldots & K_d & \ldots & H_{d-1} & H_d & H_\infty  \\
\hline
+1 & ? & \ldots & ? & \ldots & ? & ?  & ? \\
0  & +1 & \ldots & ? & \ldots & ? & ?  & ? \\
 &  &  &  & \vdots &  &  &    \\
0  & 0 & \ldots & ? & \ldots & +1 & ? & ?  \\
0  & 0 & \ldots & +1 & \ldots & 0 & +1 & ?  \\
0  & 0 & \ldots & 0 & \ldots &  0 & 0 & +1 \\
\end{array}
\right)  \;\; \Longrightarrow \;\; \left( \begin{array}{c}
Y \\ \hline 1 \\ 2 \\ \vdots \\ d-1 \\ d \\ \infty
\end{array} \right).$$
\end{comment}
$$ \begin{pNiceArray}{cccc|cccc}[first-row]
K_1 & K_2 & \cdots & K_d & \cdots & H_{d-1} & H_d & H_\infty  \\
+1 & ? & \cdots & ? & \cdots & ? & ?  & ? \\
0  & +1 & \cdots & ? & \cdots & ? & ?  & ? \\
 &  &  &  & \vdots &  &  &    \\
0  & 0 & \cdots & ? & \cdots & +1 & ? & ?  \\
0  & 0 & \cdots & +1 & \cdots & 0 & +1 & ?  \\
0  & 0 & \cdots & 0 & \cdots &  0 & 0 & +1 
\end{pNiceArray} 
\;\; \Longrightarrow \;\; \begin{pNiceMatrix}[first-row]
\vc{Y} \\ 1 \\ 2 \\ \vdots \\ d-1 \\ d \\ \infty
\end{pNiceMatrix} .$$
\caption{Schematic illustration of test matrix and outcomes sorted into block form. Here a $0$ represents a submatrix of all zeroes, a $+1$ represents a submatrix which has at least one entry equal to $1$ in each column, and $?$ represents a submatrix which could be of any form. The defective items are sorted by level to the left of the vertical line. The column labels above the matrix represents the number of elements of each type; the vector represents the outcomes of the test. \label{fig:sorted}}
\end{center}
\end{figure}

The algorithms we describe in Section \ref{sec:algo} will be effective for a variety of matrix designs. However, as in \cite{aldridge2014}, in the theoretical analysis in Sections \ref{sec:analysis_COMP}--\ref{sec:DDconverse} we assume that the matrix $\vc{x}$ is sampled according to a Bernoulli design with parameter $p$; that is, that the elements $x_{ti}$ are equal to $1$ independently of one another with a fixed probability $p$.  As in \cite[Section 2.1]{aldridge2019}, we consider a probability $p = \new/K$ for some fixed $\new$. 

In fact, as justified in \cite{aldridge2014} and in Section \ref{sec:simulation} it is often reasonable to take $\new=1$. It remains possible that some other choice may be better in some situations, although simulation evidence in Figure \ref{fig:diffp} shows that the performance of our algorithms is relatively robust to choices of $\new$ close to 1. This means that while in theory we need to know the number of defective items in the sample to design the matrix, for practical purposes it is enough to have a good estimate of this number.

The paper \cite{johnson2019} proves that performance is improved in the classical case when using matrix designs with near-constant column weights $L = \lfloor \nu T/K \rfloor $, and simulation evidence in Figure \ref{fig:diffdesigns} suggests that the same might well be true in the tropical case. However the analysis involved in \cite{johnson2019} is significantly more complicated than that in \cite{aldridge2014}, so here we restrict ourselves to the Bernoulli case for the sake of simplicity of exposition, and leave alternate matrix designs for future work.

\section{Description of tropical algorithms} \label{sec:algo}

\subsection{General remarks concerning algorithms}
In this section, we describe three algorithms which estimate the true vector of defectivity levels $\vc{U}$, given the test design matrix $\vc{x}$ and the vector of test outcomes $\vc{Y}$. These are the tropical COMP, tropical DD and tropical SCOMP algorithms, adapted from the classical algorithms of the same names in \cite{chan2011,aldridge2014} (see also \cite[Chapter 2]{aldridge2019} for a more detailed description).

We first define what is meant by an algorithm in this setting.

\begin{definition} A decoding (or detection) algorithm is a function $\wh{\vc{U}}: \{0,1\}^{T \times N} \times \Dset^T \to \Dset^N$ which estimates the defectivity level of each of the items, based only on knowledge of the test design $\vc{x}$ and outcomes $\vc{Y}$.
\end{definition}

We write $\Perr$ for the error probability of an algorithm, and $\Psuc = 1- \Perr$ for the success probability. We define 
\begin{equation} \label{eq:success}
\Perr = \pr( \wh{\vc{U}} \neq \vc{U})
\end{equation}
to be the probability that the algorithm fails to recover all the defectivity levels exactly, where the randomness comes through the design of the matrix and the value of $\vc{U}$ itself. Sometimes for emphasis we will include the name of the algorithm and the number of tests, for example by writing $\PERR{DD}{T}$.

Recovering $\vc{U}$ exactly represents a strong success criterion for this problem. For example, in some cases, we might be happy to simply recover the defective set $\KK = \{ i: U_i < \infty \}$. We later show that recovering $\vc{U}$ and recovering $\KK$ represent equivalent success criteria for tropical DD and tropical SCOMP, but not for tropical COMP. From a clinical point of view, since lower Ct levels are generally associated with higher infectiousness \cite{rao2020narrative}, it might be sufficient to recover all the items with defectivity level below a certain threshold $t$, that is to find $\bigcup_{r < t} \KK_r = \{ i: U_i < t \}$.

In this setting, we say that a \emph{false positive error} is an error of underestimating the defectivity level $U_i$ of an item $i$ -- that is, of setting $\wh{U}_i < U_i$ -- and a \emph{false negative error} is an error of overestimating the defectivity level $U_i$ of an item $i$ -- that is, of setting $\wh{U}_i > U_i$.

In the remainder of this section, we define the tropical COMP (Subsection \ref{sec:tropCOMP}), tropical DD (Subsection \ref{sec:tropdd}) and tropical SCOMP (Subsection \ref{sec:tropSCOMP}) algorithms as tropical equivalents of their established classical counterparts. All of these algorithms are relatively simple: they do not require exhaustive search over possible values of $\vc{U}$ (in contrast to the classical SSS algorithm \cite{aldridge2014}, for example), can be implemented with a small number of passes through the data, and require an amount of storage which is proportional to the number of items and tests. 

Despite this simplicity, in the classical case, the DD algorithm has performance close to optimal for certain parameter ranges. This can be seen by comparing \cite[Eq.~(1.1), (1.2)]{coja-oghlan20}, which show that DD under a constant column weight design achieves an asymptotic performance which matches that achievable by any algorithm and any test design in the case where $K \aequi N^\alpha$ and $1/2 \leq \alpha < 1$.

Also, note that while simulations show that classical SCOMP outperforms classical DD for a range of finite size problems, Coja-Oghlan {\em et al.} \cite{coja-oghlan20a}\ prove that it requires the same number of tests in an asymptotic sense, with SCOMP having the same rate (in the sense of \cite{aldridge2014}) as classical DD.

%But first, we briefly state a useful lower bound on the number of tests required for any test strategy.

\subsection{Counting bounds}

For classical group testing, a lower bound on the number of tests required is given by the  so-called `magic number' $\Tmag{class} := \log_2 \binom{N}{K}$, which can be justified on information-theoretic grounds. In fact below this number of tests there is exponential decay in performance of any algorithm, adaptive or non-adaptive, and for any test strategy. Specifically, \cite[Theorem 3.1]{baldassini2013} shows that if $T$ tests are used then in any scenario the success probability for classical group testing satisfies
\begin{equation} \label{eq:countingbound}
\Psuc \leq 2^{-(\Tmag{class}-T)} = \frac{2^T}{\binom{N}{K}}, \end{equation}
sometimes referred to as the counting bound. 

It may not be \emph{a priori} obvious how the difficulty of the tropical decoding problem with success criterion \eqref{eq:success} compares with the corresponding classical problem. In the tropical setting, we receive more information from each test through the more diverse test outcomes, which suggests the problem could be easier; but we also need to recover more information (to find the levels $\vc{U}$),  which suggests the problem could be harder. Nonetheless, if for given parameters any tropical algorithm can demonstrate performance exceeding the classical counting bound \eqref{eq:countingbound} then we can be sure that the corresponding tropical problem is easier than its classical counterpart.

By closely mimicking the proof of the classical counting bound \eqref{eq:countingbound} given in \cite{baldassini2013} we can prove its tropical counterpart.

\begin{theorem} \label{thm:counting}
Write $\Tmag{trop} := \log_{d+1} \binom{N}{\vc{K}}$, where 
\[ \binom{N}{\vc{K}} = \binom{N}{K_1, K_2, \dots, K_d, N-K} = \frac{N!}{K_1! K_2! \cdots K_d! (N-K)!} \]
is the multinomial coefficient. Then
\begin{equation} \label{eq:countingbound2}
\Psuc \leq (d+1)^{-(\Tmag{trop}-T)} = \frac{(d+1)^T}{\binom{N}{\vc{K}}}. \end{equation}
\end{theorem}

\begin{proof} See Appendix \ref{sec:counting}. \end{proof}
Writing $\binom{K}{\vc{K}} = K!/(K_1! K_2! \ldots K_d!)$ and $H(\vc{\theta}) = -\sum_{i=1}^d \theta_i \log_2 (\theta_i)$, we expand
\begin{equation}
\Tmag{trop} =  \log_{d+1} \binom{N}{K} + \log_{d+1} \binom{K}{\vc{K}} \aequi \frac{\Tmag{class}}{\log_2(d+1)} + K \frac{H(\vc{\theta})}{\log_2(d+1)}. \label{eq:tropmag}
\end{equation}
Compared with the classical case, the scaling factor $1/\log_2(d+1) < 1$ on the first term of \eqref{eq:tropmag} represents the fact that we potentially gain more information through each test, while the second additive term represents the extra information we are required to recover.

\subsection{Tropical COMP} \label{sec:tropCOMP}

We now describe the tropical COMP algorithm, which extends the classical COMP algorithm described in \cite{chan2011} (see also \cite{chan2014}) -- although the idea of the algorithm dates back at least to the work of Kautz and Singleton \cite{kautz}.

We first describe the classical COMP algorithm, which simply declares any item that appears in a negative test as non-defective. All other items are declared defective. In the notation of this paper, classical COMP can be described in the following way. For each item $i$ with $\mu_i = \infty$, we set $\wh U_i = \infty$; otherwise, $\mu_i = 1$ and we set $\wh U_i = 1$. In other words, we set $\wh U_i = \mu_i$ for each item $i$. 

The same rule $\wh U_i = \mu_i$ can also be used in tropical group testing. This is what we call the tropical COMP algorithm.

\begin{algorithm}[H]
\SetAlgoLined
\KwIn{Test design matrix $\vc{x}$ and vector of test outcomes $\vc{Y}$}
\KwOut{Estimated vector of defectivity levels $\wh{\vc{U}}$}
\lFor{\textnormal{each item} $i$}{set $\wh{U}_i = \mu_i$}
\caption{Tropical COMP algorithm}
\end{algorithm}

While both classical and tropical COMP mark items appearing in negative tests as non-defective, the tropical COMP algorithm further classifies items into estimated defectivity levels. Note that the two algorithms operate identically when $d=1$, and have some analogous properties in general. To aid terminology, we first define the notion of unexplained tests in this setting. 

\begin{definition} \label{def:explain} %[Unexplained test]
Fix a test matrix $\vc{x}$ and an estimate $\wh{\vc{U}}$ of $\vc U$. Write
  \[ \wh{Y}_t = \min_i \{ \wh U_i: x_{ti} = 1 \} \]
to be the outcome of test $t$ using matrix $\vc{x}$ if the true defectivity vector were equal to $\wh{\vc{U}}$. We say that test $t$ is \emph{unexplained} by $\wh{\vc{U}}$ if $\wh{Y}_t \neq Y_t$, where $Y_t$ is the actual test outcome, or \emph{explained} if $\wh{Y}_t = Y_t$.

We call an estimate vector $\wh{\vc{U}}$ a \emph{satisfying vector} if it explains all $T$ tests.
\end{definition}

The terminology `satisfying vector' here is the tropical group testing equivalent of the classical group testing notion of a satisfying set \cite{aldridge2014, aldridge2019}. For classical COMP, the estimate given is a satisfying set \cite[Lemma 2.3]{aldridge2019}) -- indeed, the largest satisfying set. 
We have a similar result for tropical COMP.

\begin{lemma} \label{rmk:COMPsatisfying} 
The estimate $\wh{\vc{U}}^{\mathrm{COMP}}$ given by tropical COMP is a satisfying vector.

Further, $\wh{\vc{U}}^{\mathrm{COMP}}$ is the least satisfying vector, in that if $\vc{V} \in \Dset^N$ is also a satisfying vector, then $U^{\text{COMP}}_i  \leq V_i$ for all items $i$.
\end{lemma}

\begin{proof}
For the first part, take an arbitrary test $t$ with outcome $Y_t$. All items $i$ included in this test have $U_i \geq \mu_i \geq Y_t$. Further, there must be an item $j$ with $U_j = Y_t$, otherwise the test outcome would be greater than $Y_t$. For that item, $\mu_j = Y_t$. Hence, 
\[ \wh{Y}_t = \min_i \{\mu_i : x_{ti} = 1\} = \mu_j =  U_j = \min_i \{U_i : x_{ti} = 1\} = Y_t , \]
and the test is explained. Since the test $t$ was arbitrary, we have $\wh{\vc{Y}} = \vc{Y}$, and hence $\wh{\vc{U}}^{\mathrm{COMP}}$ explains all the tests. 

For the second part, note that any satisfying vector $\vc{V}$ must have $V_i \geq \mu_i = U^{\text{COMP}}_i$ for all $i$. To see this, consider a vector $\vc{V}$ and item $j$ with $V_j < \mu_j$. Then let $t$ be a test containing item $j$ for which $Y_t = \mu_j$. There must be at least one such test, by the definition of $\mu_j$, unless $j$ is never tested. If $j$ is never tested, then by assumption $V_j \in \Dset$ has $V_j \geq 1 = \mu_j$. For this test $t$,
\[ \min_i \{ V_i : x_{ti} = 1 \} \geq \mu_j > V_j , \]
so $\vc{V}$ is not satisfying.
\end{proof}

We know that classical COMP never makes false negative errors. The same is true for tropical COMP -- recall that we use this terminology to refer to an error of the form $\wh{U}_i > U_i$.

\begin{lemma} \label{rmk:COMPfalse-}
Tropical COMP never makes false negative errors.
\end{lemma}

\begin{proof}
This follows directly from Lemma \ref{lem:deduct}, which tells us that $U_i \geq \mu_i$, where $\mu_i$ is the tropical COMP estimate.
\end{proof}

For tropical COMP, the success criterion given by \eqref{eq:success} to recover the whole vector $\vc{U}$ is not equivalent to the success criterion of merely recovering the defective set $\KK$. It is true that if any algorithm correctly recovers $\vc{U}$, so that $\wh{\vc{U}}=\vc{U}$, then it also recovers $\KK$ as $\wh{\KK} = \{ i : \wh{U}_i < \infty\} = \KK$. But the following example shows that for tropical COMP the converse does not hold; that is, just because COMP fails to recover $\vc{U}$, that does not necessarily mean it also fails to recover $\KK$:

\begin{example}
    Suppose we have two items, with true defectivity levels $\vc{U} = (1, 2)$. Suppose further that we run just one test, which contains both items, so $\vc{x} = (1, 1)$ and $\vc{Y} = (1).$ Then both items are in just one test with outcome $Y_1 = 1$, so have $\mu_1 = \mu_2 = 1$. Tropical COMP therefore incorrectly estimates $\wh{\vc{U}} = (1, 1) \neq \vc{U}$. However, it does succeed in recovering the defective set $\wh{\KK} = \KK = \{1,2\}$. 
\end{example}

Despite this, we show in Section \ref{sec:analysis_COMP} that tropical COMP asymptotically requires the same number of tests to recover $\vc{U}$ as classical COMP does to recover $\KK$.

\subsection{Tropical DD} \label{sec:tropdd}
We now describe the tropical DD algorithm. This extends the classical DD algorithm introduced in \cite{aldridge2014}, which works in three steps:
\begin{enumerate}
    \item Every item appearing in a negative test is non-defective. (All other items are `possibly defective'.)
    \item If a positive test contains a single possibly defective item, that item is defective.
    \item All remaining items are assumed non-defective.
\end{enumerate}
Tropical DD works the same way, except in step 2, it takes account of the different levels of defectivity in the tropical testing. Recalling that a $\PD(r)$ item is one with $\mu_i = r$, the tropical DD algorithm is as follows:

\begin{algorithm}[H]
\SetAlgoLined
\KwIn{Test design matrix $\vc{x}$ and vector of test outcomes $\vc{Y}$}
\KwOut{Estimated vector of defectivity levels $\wh{\vc{U}}$}
\lFor{\textnormal{each item} $i$ \textnormal{with} $\mu_i = \infty$}{set $\wh{U}_i = \infty$}
\For{\textnormal{each test} $t$ \textnormal{with} $Y_t = r < \infty$}{
  \lIf{\textnormal{there exists only one} $\PD(r)$ \textnormal{item} $i$ \textnormal{in test} $t$}{\textnormal{set} $\wh{U}_i = r$}}
Declare all remaining unclassified items to have $\wh U_i = \infty$\;
\caption{Tropical DD algorithm} \label{alg:tropiDD}
\end{algorithm}

To understand why this algorithm works, consider a test $t$ with outcome $Y_t = r$. Observe that (by Definitions \ref{def:outcome} and \ref{def:mudef} respectively): 
\begin{enumerate}
\item [(a)] test $t$ cannot contain any items $i$ with $U_i < r$, and must contain at least one `special item' $j$ with $U_j = r$;
\item [(b)] every item $i$ appearing in test $t$ has $\mu_i \geq r$.
\end{enumerate}
Suppose that (apart from $j$) all the other items $i$ in test $t$ have $\mu_i > r$, so none of them are in $\PD(r)$. Then we know (by Lemma \ref{lem:deduct}) that each such item has $U_i \geq \mu_i > r$, and cannot be a special item. Hence the remaining item $j$ must be the special item that we seek. In other words, the sole $\PD(r)$ item in the test is marked as definitely defective at level $r$. This mirrors the classical case where if there is a single $\PD$ item in a test, it is marked as definitely defective.

We can think of the problem of finding $\KK$ in the classical case as now being split into sub-problems of finding $\KK_1, \ldots, \KK_d$ in the tropical case. 
It is helpful to think of moving upwards through the rows in the block formulation of Figure \ref{fig:sorted}: 
\begin{enumerate}
\item By examining the tests with outcome $\infty$, we can identify $H_\infty$ non-defective items and remove them from consideration for tests of outcome $r < \infty$. 
\item In general, for $r=d, d-1, \ldots, 1$, by examining all the tests with outcome $r$, we hope to find all the defective items $i(r,1), \ldots, i(r,K_r)$ and to find the $H_r$ non-defective items that are in $\PD(r)$ and remove them from consideration for tests of outcome lower than $r$.
\end{enumerate}

We note that the operation of classical DD is the same as the operation of tropical DD when $d=1$. We know that classical DD never makes false positive errors \cite[Lemma 2.2]{aldridge2019}. The same is true for tropical DD:

\begin{lemma} \label{rmk:DDnofalse+} 
Tropical DD never makes false positive errors. Indeed the only errors it can make is wrongly declaring a defective items of some finite level $U_i = r$ to be non-defective $\wh U_i = \infty$.
\end{lemma}

\begin{proof} 
The first step finds some non-defective items from negative tests, and so is error-free. The second step identifies the sole defective item that can explain the outcome of the test it is in. It is thus also error-free. The final step is the only point at which errors can be made; specifically, false negative errors where a defective item is marked non-defective can occur. 
\end{proof}

For tropical DD, the success criteria of recovering the vector $\vc{U}$ and recovering the defective set $\KK$ are equivalent. We know that if an algorithm recovers $\vc{U}$, then it recovers $\KK$. To prove equivalence of the success criteria, it suffices to show that if tropical DD fails to recover $\vc{U}$, then it fails to recovers $\KK$. This is done in the following paragraph.

Suppose that tropical DD fails to recover $\vc{U}$. Then by Lemma \ref{rmk:DDnofalse+}, the only errors that could have been made are false negative errors where a defective item is wrongly marked non-defective. Hence, tropical DD also fails to recover $\KK$.

\subsection{Tropical SCOMP} \label{sec:tropSCOMP}

We now describe the tropical SCOMP algorithm, which extends the classical SCOMP algorithm introduced in \cite{aldridge2014}. 

Classical SCOMP starts with the estimate given by the DD algorithm, but considers the fact that DD can leave some positive tests unexplained, since they can contain some  defective items whose status we are unsure about. Since each such unexplained test must contain at least one such defective, we might consider possible defective items which appear in many unexplained tests as most likely to be defective. Hence the classical SCOMP algorithm greedily adds such items to the estimated defective set $\hat{\KK}$ until all tests are explained. 

Similarly, the tropical SCOMP algorithm starts with the estimate given by tropical DD. It then greedily adds items to the sets $\hat{\KK}_r$, for each $r$ such that there are unexplained tests of outcome $r$. This is done until all tests are explained.

\begin{algorithm}[H]
\SetAlgoLined
\KwIn{Test design matrix, $\vc{x}$, and vector of test outcomes $\vc{Y}$}
\KwOut{Estimated vector of defectivity levels, $\wh{\vc{U}}$}
Initialize $\wh{\vc{U}}$ as the estimate $\wh{\vc{U}}_{\text{DD}}$ of $\vc{U}$ produced by the DD algorithm\;
\While{\textnormal{unexplained tests exist}}{
    Choose a test outcome $r$ from the unexplained tests\;
    Retrieve all the tests with outcome $r$\;
    Find the $\PD(r)$ item $i$ that occurs the most times in those tests (ties can be broken arbitrarily)\;
    Set $\wh{U}_i = r$ and update the list of unexplained tests\;
}
\caption{Tropical SCOMP algorithm}
\end{algorithm}

Note that tropical SCOMP attempts to solve the sub-problems of finding $\KK_1, \ldots, \KK_d$ that are not solved by tropical DD. The action of classical SCOMP is the same as that of tropical SCOMP when $d=1$.

\begin{remark}
    If tropical DD succeeds, then so does tropical SCOMP. This is because tropical SCOMP starts with the estimate produced by tropical DD. If tropical DD succeeds, then no tests are unexplained and tropical SCOMP also succeeds (cf.\ \cite[Theorem 2.5]{aldridge2019}). 
    % Theorem cited is about about how any rate achievable by DD is achievable by SCOMP. Remark is the analogous statement for the tropical case
\end{remark}

We show that the success criteria of recovering $\vc{U}$ and recovering $\KK$ are equivalent for tropical SCOMP. Similar to the case of tropical DD, it suffices to show that if tropical SCOMP fails to recovers $\vc{U}$, then it fails to recovers $\KK$. This is done in the following paragraph.

Suppose that tropical SCOMP fails to recover $\vc{U}$. Then necessarily, tropical DD also fails to recover $\vc{U}$. Then there exists an item $i \in \KK$ such that there are no tests in which it is the only $\PD(\mu_i)$ item. Since tropical SCOMP fails to recover $\vc{U}$, at least one such item $i$ was not chosen to explain the test outcomes of any of the tests that it is in and is marked as non-defective. Hence, $\wh{\KK} \neq \KK$.

\subsection{Comparison of tropical algorithms}

\begin{table}
\centering
\begin{tabular}{lccc}
\hline
\textbf{} & \textbf{satisfying} & \textbf{no false $\boldsymbol +$} & \textbf{no false $\boldsymbol -$} \\ \hline
COMP & $\cmark $ & $\xmark$ & $\cmark$ \\
DD & $\xmark$ & $\cmark$ & $\xmark$ \\
SCOMP & $\cmark$ & $\xmark$ & $\xmark$ \\ \hline
tropical COMP & $\cmark$ & $\xmark$ & $\cmark$ \\
tropical DD & $\xmark$ & $\cmark$ & $\xmark$ \\
tropical SCOMP & $\cmark$ & $\xmark$ & $\xmark$ \\ \hline
\end{tabular}
\caption{Summary of features of algorithms in the classical and tropical case: (i)~whether the output estimate $\wh{\vc{U}}$ is guaranteed to explain all test outcomes; (ii)--(iii)~guarantees on false positives or false negatives.}
\label{tab:algofeatures}
\end{table}

Table \ref{tab:algofeatures} summarises the features of the tropical algorithms, while comparing them to the classical algorithms (cf.\ \cite[Table $2.1$]{aldridge2019}). We now present a worked example which illustrates the operation of the various tropical algorithms:

\begin{example} \label{example:algorithms}
Suppose we use the test design $\vc{x}$ and receive the outcomes $\vc{Y}$ as follows:
\[
\vc{x} = \begin{bmatrix}
1 & 0 & 0 & 0 & 0 & 0 & 0 \\
1 & 0 & 1 & 0 & 0 & 0 & 1 \\
0 & 1 & 0 & 1 & 1 & 0 & 0 \\
0 & 1 & 0 & 0 & 1 & 1 & 0 \\
1 & 0 & 0 & 0 & 1 & 0 & 0 \\
\end{bmatrix}
\qquad
\vc{Y} = \begin{bmatrix}\infty \\ 37 \\ \infty \\ 29 \\ \infty\end{bmatrix}.
\]

It is convenient to first calculate $\vc{\mu}$. For example, item $1$ occurs in tests $1, 2, 5$. We then deduce 
\[ \mu_1 = \max_{t \in \{1,2,5\}} Y_t = \max\{\infty,37,\infty\} = \infty . \]
Proceeding similarly for the other items, we obtain
\[ \vc{\mu} = \big(\infty, \infty, 37, \infty, \infty, 29, 37 \big) . \]

\begin{description}
\item[Tropical COMP:] We set $\wh{\vc{U}} = \vc{\mu}$, obtaining the following:
\[ \wh{\vc{U}}^{\text{COMP}} = \big(\infty, \infty, 37, \infty, \infty, 29, 37 \big) . \]

\item[Tropical DD:]
 In the first step, we find the items with $\mu_i = \infty$. These are items $1, 2, 4$ and $5$. We declare these to be non-defective, so $\wh U^{\text{DD}}_1 = \wh U^{\text{DD}}_2 = \wh U^{\text{DD}}_4 = \wh U^{\text{DD}}_5 = \infty$.

In the second step, we check each positive test $t$ and look to see if they contain a single $\PD(Y_t)$ item. For test $2$, there are two $\PD(Y_2) = \PD(37)$ items in the test, items $3$ and $7$, so DD does nothing. For test $4$, items $2, 5$ and $6$ appear, but only item $6$ is a $\PD(Y_4) = \PD(29)$ item. Hence, the tropical DD algorithm sets $\wh{U}^{\text{DD}}_6 = 29$. 

Finally, in the third step, items 3 and 7, which have not yet been classified, get assigned a defectivity level of $\wh{U}^{\text{DD}}_3 = \wh{U}^{\text{DD}}_7 = \infty$.

Hence, the output of the tropical DD algorithm is:
\[ \wh{\vc{U}}_{\text{DD}} = \big(\infty, \infty, \infty, \infty, \infty, 29, \infty\big) . \]

\item[Tropical SCOMP:]
The algorithm initializes with the tropical DD estimate $\wh{\vc{U}} = \wh{\vc{U}}_{\text{DD}}.$ The corresponding outcome would be (written as the transpose, a row vector)
\[ \wh{\vc{Y}} = (\infty, \infty, \infty, 29, \infty)^\top, \]
where $\wh{Y}_2 = \infty \neq 37 = Y_2$. Hence, test $2$ is the only unexplained test. We retrieve the $\PD(37)$ items in test $2$. These are items $3$ and $7$. Because these items both appear in the same number of tests with outcome $37$, namely one, the tropical SCOMP algorithm chooses between them arbitrarily -- let's say it chooses item $7$ -- and assigns the defectivity level of $\wh{U}^{\text{SCOMP}}_7 = 37$ to it. Now no tests remain unexplained, and the algorithm terminates.

Hence the algorithm returns
\[ \wh{\vc{U}}^{\text{SCOMP}} = 
\big(\infty, \infty, \infty, \infty, \infty, 29, 37\big) . \]
\end{description}
\end{example}

\section{Simulation results} \label{sec:simulation}
In this section, we present some simulation results. We empirically compare the performance of the tropical and classical algorithms, and investigate how changing the probability $p$ and the sequence $\vc{K} = (K_1, \ldots, K_d)$ affects their performance. We also investigate the effect of using a combinatorial model with random defectivity levels for defective items, as opposed to the model with fixed $K_r$ introduced in Section \ref{sec:notation}. Finally, we compare the Bernoulli design and the near-constant column weight design, described in \cite[Section 2.7]{aldridge2019}.

\begin{figure}[ht!]
    \centering
    \includegraphics[scale=0.55]{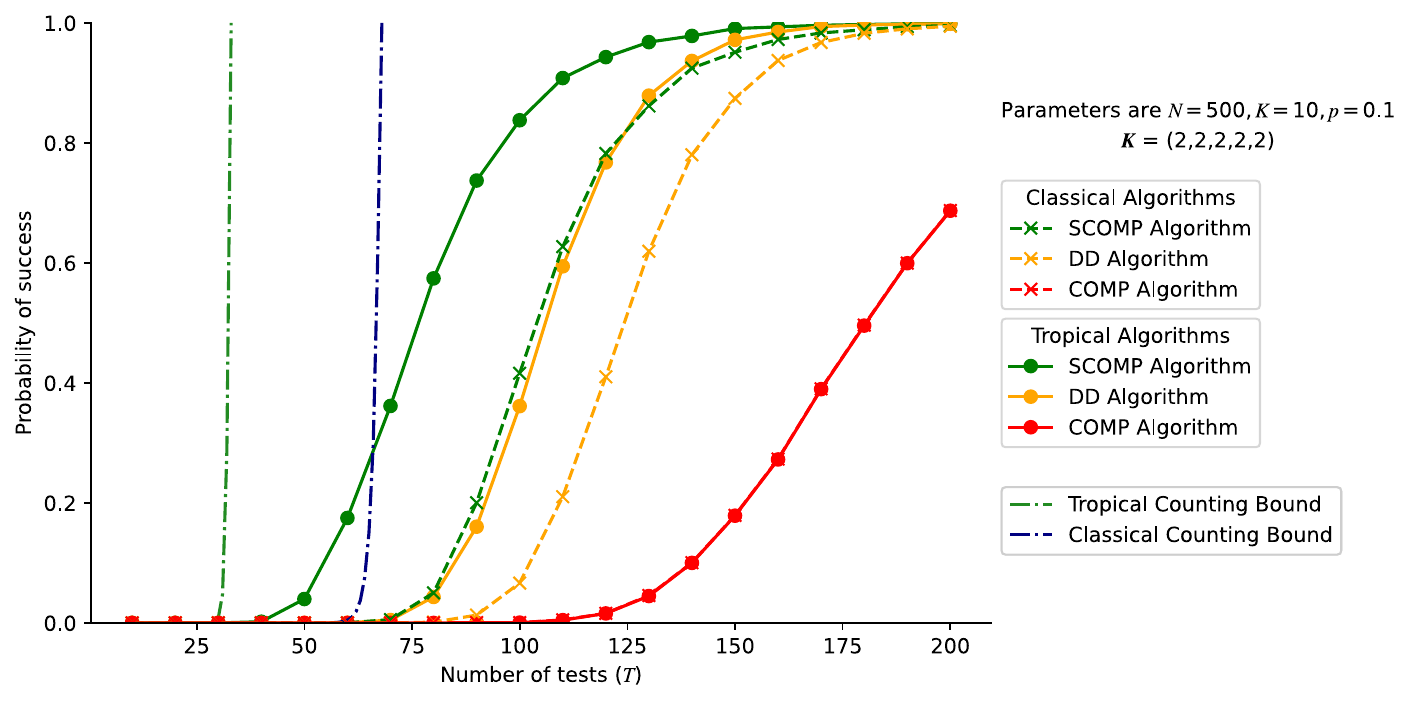}
    \caption{Empirical performance of the classical COMP, DD and SCOMP algorithms together with their tropical counterparts, through simulation with a Bernoulli design. For comparison, we plot the classical and tropical counting bounds of \eqref{eq:countingbound} and \eqref{eq:countingbound2}. The parameters chosen are $N = 500, K = 10, p = 0.1, \vc{K} = (2, 2, 2, 2, 2).$ Each point is obtained through $10^4$ simulations.}
    \label{fig:algoplots}
\end{figure}

Figure \ref{fig:algoplots} shows the performance of the tropical algorithms, relative to the performance of the classical algorithms and to the counting bounds \eqref{eq:countingbound} and \eqref{eq:countingbound2}. Figure \ref{fig:algoplots} shows for the chosen set of parameters that the tropical COMP algorithm performs almost identically to its classical counterpart (the lines are so close that they are difficult to distinguish), and the tropical DD and SCOMP algorithms perform better than their classical counterparts. We also notice that in this case tropical SCOMP beats the classical counting bound \eqref{eq:countingbound} for small values of $T$, showing that the tropical model can allow genuine performance gains over even adaptive classical group testing algorithms. 

\begin{figure}[ht!]
    \centering
    \includegraphics[scale=0.55]{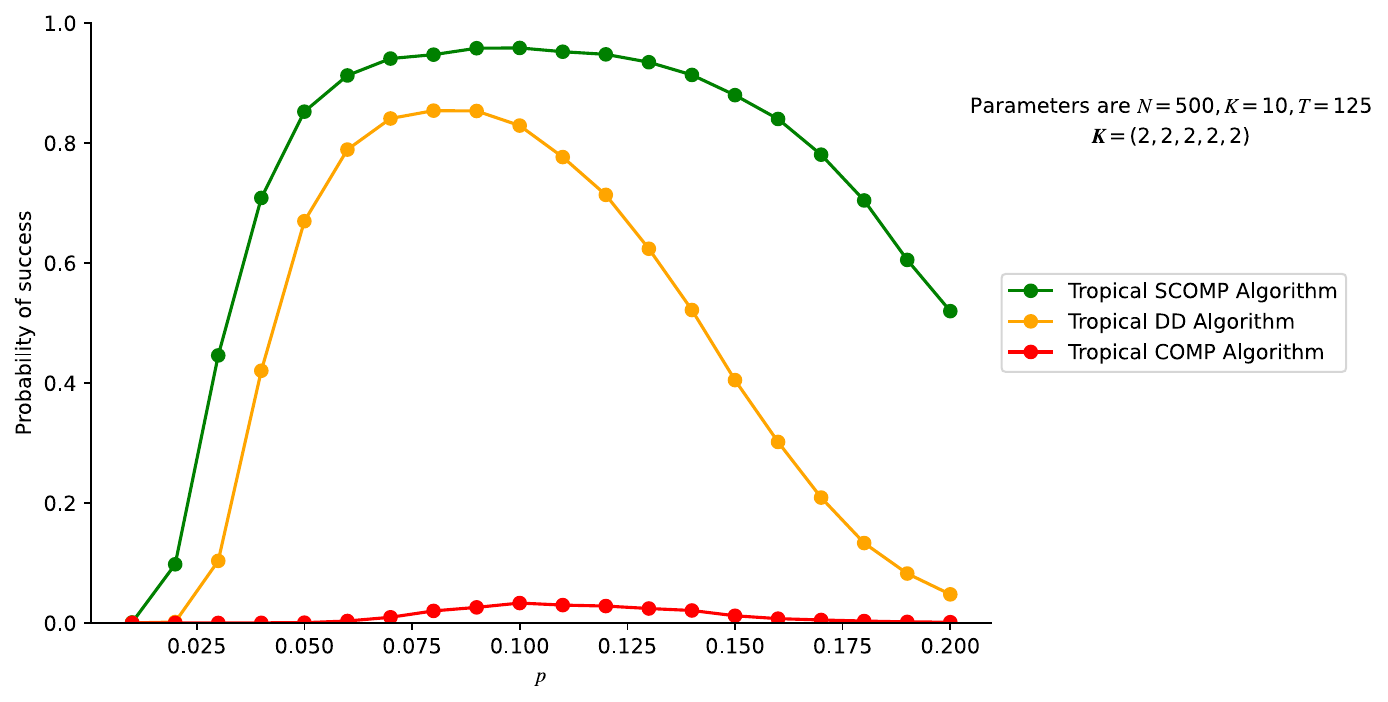}
    \caption{Simulation of the tropical COMP, DD and SCOMP algorithms with a Bernoulli design to investigate the effect of changing the parameter $p$ for the Bernoulli design. The parameters are $N = 500, K = 10, T = 125$ and $\vc{K} = (2, 2, 2, 2, 2)$. Each point is obtained through $10^4$ simulations.}
    \label{fig:diffp}
\end{figure}

Figure \ref{fig:diffp} shows how the performance of the tropical algorithms vary with $p$, for fixed $N$, $T$ and $\vc{K}$. Figure \ref{fig:diffp} % suggests that the optimal $p$ may be smaller for tropical DD and SCOMP compared to their classical counterparts. [VP: Removed since the plot is qualitatively the same for $\vc{K} = (K)$, i.e d = 1 -- classical case.]
shows that the performance of tropical DD and tropical SCOMP have a relatively wide plateau near their peak, indicating some robustness to misspecification of $p$, and showing that in general the choice $p=1/K$ is close to optimal for each algorithm.

\begin{figure}[ht!]
    \centering
    \includegraphics[scale=0.54]{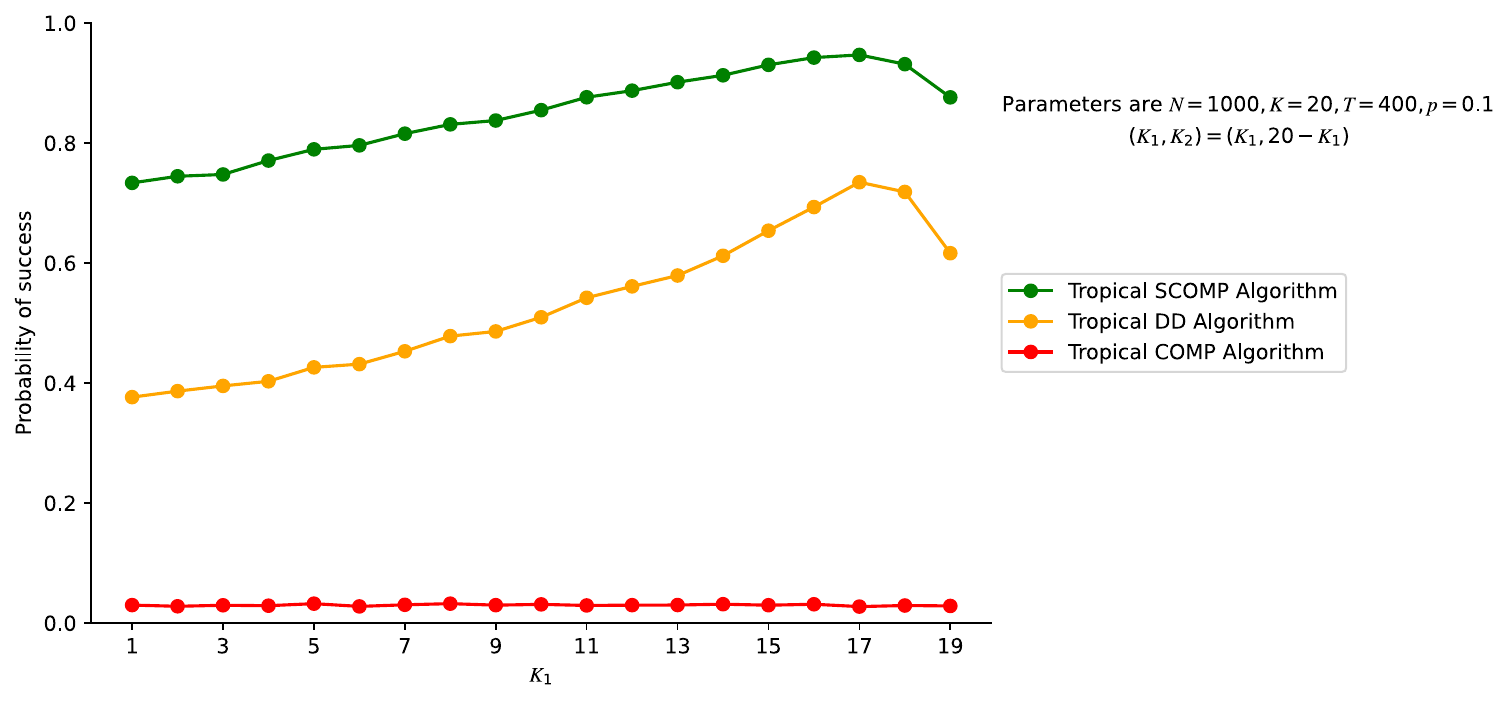}
    \caption{Simulation of the tropical COMP, DD and SCOMP algorithms with a Bernoulli design to investigate the effect of changing $\vc{K}$ on the performance. The parameters are $N = 1000, K = 20, T = 400, p = 0.1$ and $\vc{K} = (K_1, 20-K_1)$. Each point is obtained through $10^4$ simulations.} 
    \label{fig:diffx}
\end{figure}

Figure \ref{fig:diffx} shows how the performance of the tropical algorithms vary as $\vc{K}$ is varied, for fixed $N, K, T$ and $d$. We note that there are $\binom{K-1}{d-1}$ possible vectors $\vc{K}$ that sum to $K$ while having each $K_i > 0$. Also, a $d$-dimensional plot is required to visualize the performance of the algorithms for all the $\vc{K}$ simultaneously. Hence, for simplicity of exposition, we only present the case $d=2$.  Figure \ref{fig:diffx} shows that changing $\vc{K}$ has very little effect on the performance of tropical COMP. This is quantified later in Section \ref{sec:analysis_COMP}, where we find that, for tropical COMP, the error contribution of the $K$ defective items is small compared to that of the $N-K$ non-defective items. Figure \ref{fig:diffx} also shows that the performance of tropical DD and tropical SCOMP improves as $K_1$ increases, until reaching a peak. 

\begin{figure}[ht!]
    \centering
    \includegraphics[scale=0.54]{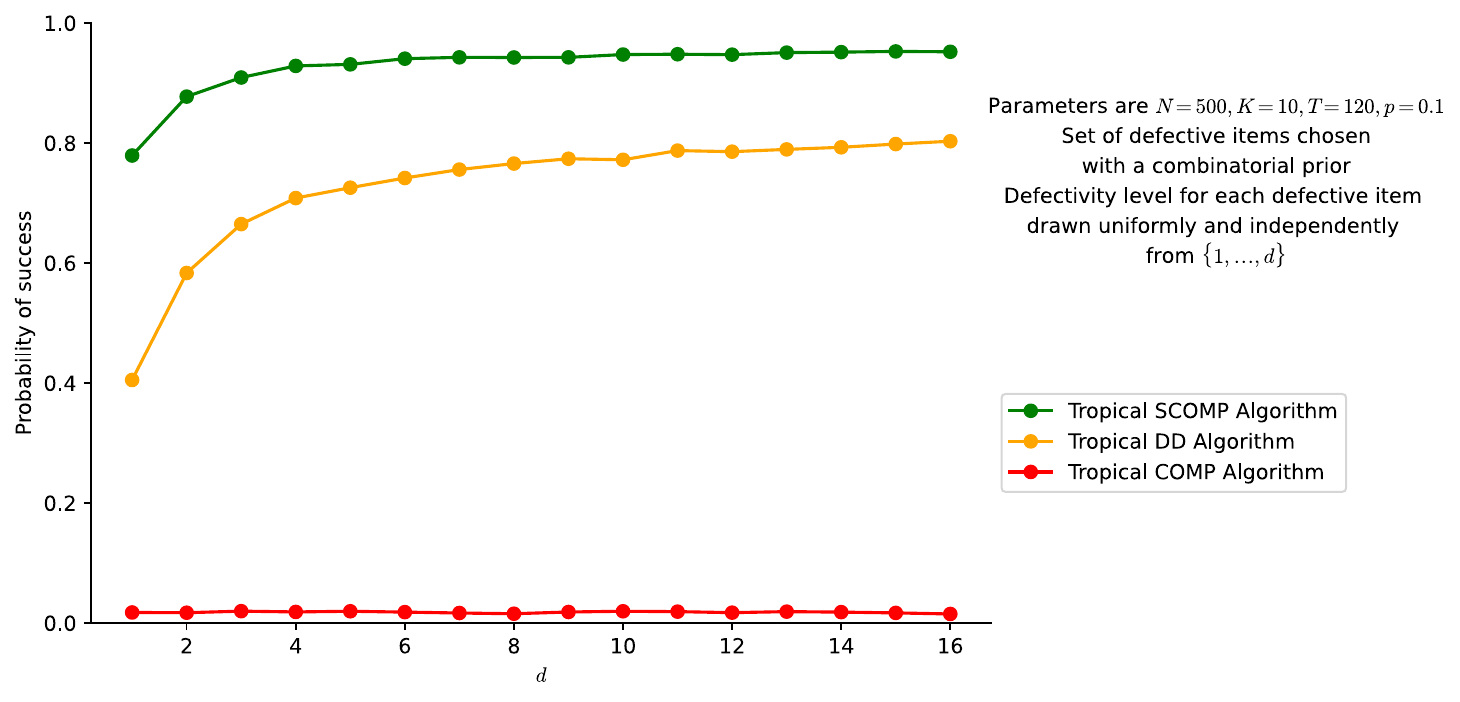}
    \caption{Simulation of the tropical COMP, DD and SCOMP algorithms to investigate their performance with a Bernoulli design when the set of defective items, $\KK$, is chosen with a combinatorial prior, and the defectivity level for each defective item is drawn uniformly and independently from $\{1, \ldots, d\}$. The parameters are $N = 500, K = 10, T = 120, p = 0.1$. Each point is obtained through $10^4$ simulations.} 
    \label{fig:diffd}
\end{figure}

Figure \ref{fig:diffd} shows the effect on the performance of the tropical algorithms when the defective set $\KK$ is chosen with a combinatorial prior, and the defectivity level for each defective item is drawn uniformly and independently from $\{1, \ldots, d\}$. We note that that the performance of tropical DD and tropical SCOMP improves as $d$ increases, until reaching a peak, while the performance of tropical COMP does not change.

\begin{figure}[ht!]
    \centering
    \includegraphics[scale=0.52]{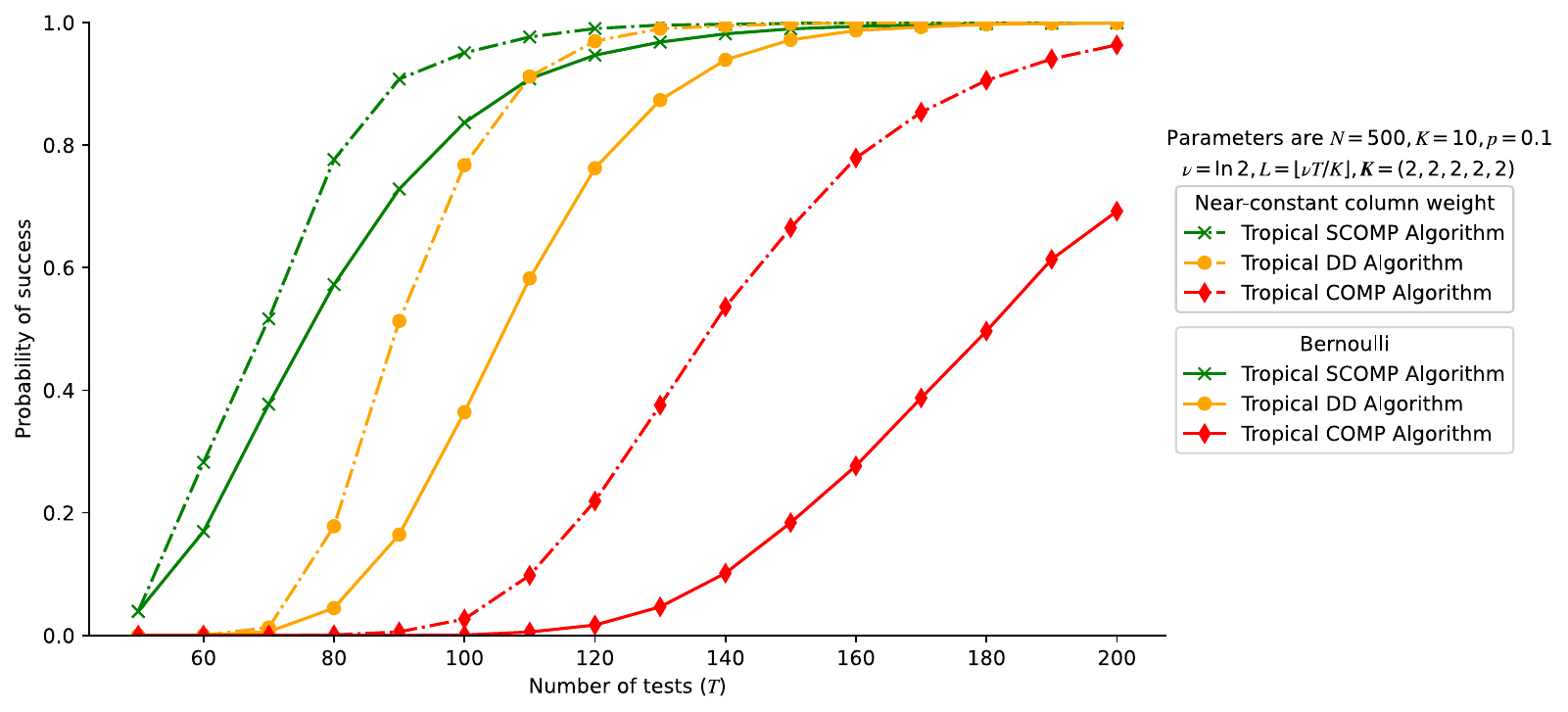}
    \caption{Simulation of the tropical COMP, DD and SCOMP algorithms to investigate their performance with a Bernoulli design as well as with a near-constant column weight design. The parameters are $N = 500, K = 10, p = 0.1, \nu = \ln 2,  L = \lfloor \nu T/K \rfloor, \vc{K} = (2,2,2,2,2)$. Each point is obtained through $10^4$ simulations.} 
    \label{fig:diffdesigns}
\end{figure}
Finally, Figure \ref{fig:diffdesigns} compares the performance of the tropical algorithms with the Bernoulli design and with a near-constant column weight design. We see that the performance mirrors the classical case (cf. \cite[Figure 2.3]{aldridge2019}).

%\clearpage
\section{Analysis of tropical COMP algorithm} \label{sec:analysis_COMP}

In this section, we give an analysis of the performance of tropical COMP. The main result of this section (Theorem \ref{thm:COMPtests}) shows that the number of tests needed to ensure a vanishing error probability using a Bernoulli design is asymptotically identical to that needed in the classical case. 

\subsection{Achievability result}

Our main result for tropical COMP is the following (cf.\ \cite[Eq.~(2.9)]{aldridge2019}):

\begin{theorem}
    Let $\delta >0$. Let $p = \new/K$, for $0 < \new < K$. Taking
    \[ T \geq (1+\delta)\,\frac{\ee^{\new}}{\new} K \ln N \]
    ensures that the tropical COMP error probability $\Perr$ is asymptotically at most $N^{-\delta}$. \label{thm:COMPtests}
\end{theorem}

\begin{remark}
    Note that $T = (1+\delta) \frac{\ee^{\new}}{\new} K \ln N$ is minimised over $\new$ when $\new = 1$. This corresponds to choosing the same optimal $p$ as in the classical case. We note that tropical COMP, similar to classical COMP, is reasonably robust to misspecification of $p$ (cf.\ \cite[Remark 2.3]{aldridge2019}).
\end{remark}

To reach the result of Theorem \ref{thm:COMPtests}, we find a bound on the error probability of tropical COMP using a Bernoulli design. This bound, given below, extends the corresponding bound by Chan {\em et al.}\ for classical COMP, given in \cite[Eq.~(8)]{chan2011}, to the tropical setting.

\begin{lemma} \label{lem:COMP bound}
For a Bernoulli test design with parameter $p$, we can bound the error probability of Tropical COMP from above by
\begin{equation} \label{eq:COMPbd}
\PERR{COMP}{T} \leq \sum_{r \in \Dset} K_r (1- p(1-p)^{\sum_{i < r}K_i})^T.
\end{equation}
\end{lemma}

\begin{proof}
To obtain an upper bound on the error probability of tropical COMP, we consider each item in turn, using the fact that the union bound
\begin{equation} \label{eq:COMPunion}
 \Perr = \pr \left( \bigcup_i \{ \wh{U}_i \neq U_i \} \right) \leq \sum_{i} \pr( \wh{U}_i \neq U_i),\end{equation}
tells us that we only need to control the individual probabilities that an item is misclassified.

Any given item $i$ with $U_i = r$ is misclassified only if it is intruding. This happens if every test which it appears in contains at least one of the $\sum_{i <r} K_i$ items of lower level.  For a given test, the chance that it contains $i$ but doesn't contain such an item is $p (1-p)^{\sum_{i < r} K_i}$. Hence using independence between tests, we have that
\begin{equation} \label{eq:COMPind}
\pr( \wh{U}_i \neq U_i) \leq (1- p(1-p)^{\sum_{i < r} K_i})^T.
\end{equation}
The result follows on substituting \eqref{eq:COMPind} in \eqref{eq:COMPunion}.
\end{proof}

We can now prove the main result.

\begin{proof}[Proof of Theorem \ref{thm:COMPtests}]
    This proof is adapted from the one given in the classical case by Chan {\em et al.}\ in \cite{chan2011}. Let $T = \beta K \ln N$. The key is to observe for a given $T$ and $p$ that the function $f(\ell) = (1- p(1-p)^{\ell})^T$ is increasing in $\ell$. 

    Hence we can write \eqref{eq:COMPbd} as 
    \begin{equation} \label{eq:sumrep}
     \Perr \leq \sum_{r \in \Dset} K_r f \left( \sum_{i <r} K_i \right) \leq \sum_{r \in \Dset} K_r f \left(  K \right)    = N f(K).\end{equation}   
    Then, setting $p = \new/K$ in Lemma \ref{lem:COMP bound}, we have 
\begin{align*} 
\Perr &\leq N \exp(-Tp(1-p)^{K}) \\
&= N \exp(-\beta \new (1 - \new/K)^K \ln N) \\
&\aequi N \exp(-\beta \new \ee^{-\new} \ln N) \qquad \text{as } K \to \infty \\
&= N^{1- \beta \new \ee^{-\new}}.
\end{align*}
Hence, taking $\beta = (1+\delta) \frac{\ee^{\new}}{\new}$ ensures that $\Perr$ is asymptotically at most $N^{-\delta}$. 
\end{proof}

\subsection{Contributions to the error probability}

\begin{figure} 
    \centering
    \includegraphics[scale=0.64]{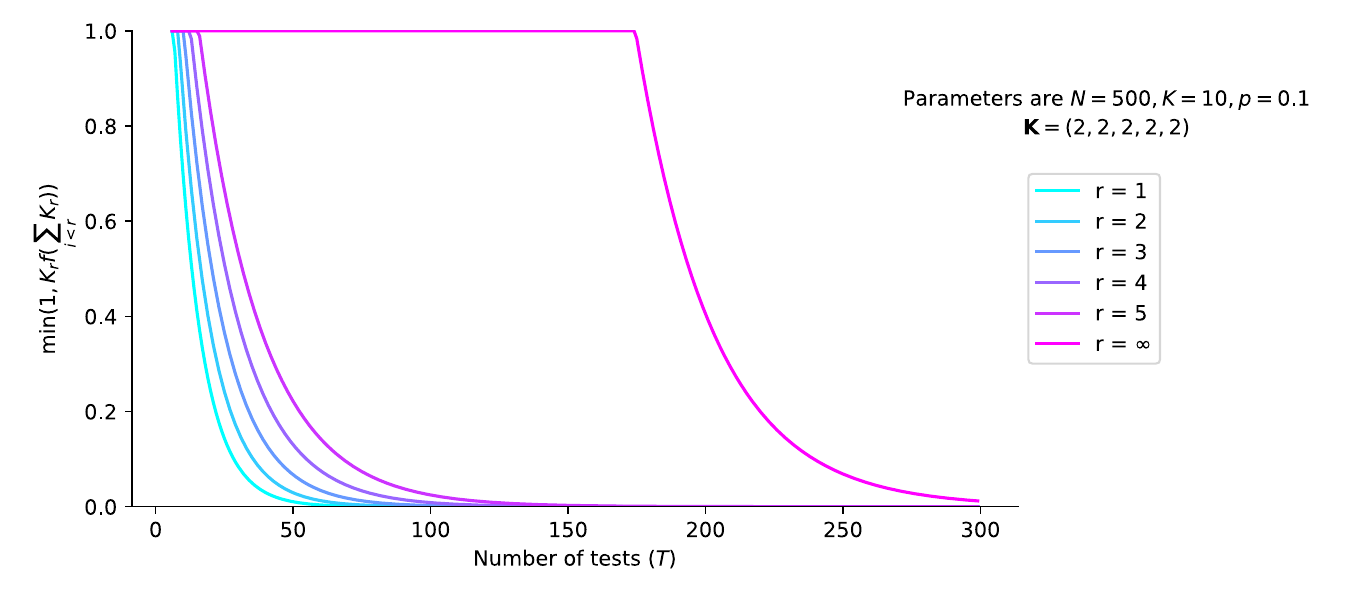}
    \setlength{\abovecaptionskip}{-10pt}
    \caption{Plot illustrating the variation of $\min\{1,K_r f(\sum_{i<r}K_r)\}$ with $T$, for each $r \in \Dset$. The parameters chosen are $N = 500, K= 10, p = 0.1$ and $\vc{K} = (2, 2, 2, 2, 2)$.} \label{fig:COMPsummands}
\end{figure}

Figure \ref{fig:COMPsummands} illustrates the contribution of each summand to the bound \eqref{eq:COMPbd} for a range of values of $T$. It is clear that the dominant term contributing to the error bound is $r = \infty$ (that the dominant error event is for a non-defective item to be wrongly classified as defective, and that the defective items are more typically correctly classified).

Indeed, \eqref{eq:sumrep} implies that the proportion of the bound \eqref{eq:COMPbd} arising from the $r=\infty$ term is at least $1-K/N$, since this result gives 
$$ \frac{(N-K) f(K)}{\sum_{r \in \Dset} K_r f \left( \sum_{i < r} K_i \right)} \geq \frac{(N-K) f(K)}{N f(K)} = 1 - \frac{K}{N} .$$ 

In fact in the `uniform' case where $\vc{K} = (K/d, \ldots, K/d)$, the contributions to the bound
\begin{equation} \label{eq:doubleexp}
K_r f \left(\sum_{i < r} K_i \right) \aequi K_r \exp \left( - T p(1-p)^{(r-1) K/d} \right) \aequi \frac{K}{d} \exp \left( - T p \ee^{ - p(r-1) K/d} \right) \end{equation} 
decay doubly-exponentially as $r$ gets smaller, meaning that the errors are overwhelmingly likely to arise from wrongly classifying items with high $r$.

\subsection{Error probability for different defectivity sequences}

For a fixed number of defective items $K$, it would be interesting to know what defectivity sequences $\vc{K}$ make the tropical group testing problem hardest or easiest. This is explored via simulation in Figure \ref{fig:diffx}, but we would also like to understand which sequences $\vc{K}$ lead to the largest and smallest error probability for tropical COMP. Unfortunately, we cannot directly control the error probability in this way. However, we can use the COMP error bound \eqref{eq:COMPbd} to induce a partial order on sequences $\vc{K}$ with the same sum. This will show that the error bound is smallest in the classical case where all items have the same level and largest in the case where each item has distinct levels.

Given a sequence $\vc{K} = (K_1, \ldots, K_d)$, we can sort the items in increasing order of level, and for each item $k$ write $L_k$ for the number of items with a strictly lower level. For example, with $K = 8$, the sequence $\vc{K}^{(1)} = (2,2,2,2)$ induces the sequence $\vc{L}^{(1)} = (0,0,2,2,4,4,6,6)$, while the sequence $\vc{K}^{(2)} = (1,1,1, \ldots, 1)$ induces the sequence $\vc{L}^{(2)} = (0,1,2,3,4,5,6,7)$.

We can compare two sequences $\vc{K}^{(i)} = \big( K_1^{(i)}, \ldots, K_d^{(i)}\big)$, for $i=1,2$, and define a partial order $\vc{K}^{(1)} \preceq \vc{K}^{(2)}$ if the corresponding sequences $L_k^{(1)} \leq L_k^{(2)}$ for all $k$.  Hence in the example above, $\vc{K}^{(1)} \preceq \vc{K}^{(2)}$. In this partial ordering, for a given $K$ the minimal sequence will be $\vc{K} = (K)$, giving $\vc{L} = (0,0, \ldots, 0)$, and the sequence $\vc{K} = (1,1, \ldots, 1)$ as seen above will be maximal.

Now, note that the bound on the RHS of \eqref{eq:COMPbd} respects this partial order. That is, since the $r = \infty$ term will be the same for all such sequences, we can regard the variable part of the bound \eqref{eq:sumrep} as a sum over defective items
\begin{equation} \label{eq:COMPbd2}
 \sum_{k \in \KK} f(L_k),\end{equation}
and use the fact that the function $f(\ell)$ is increasing in $\ell$ to deduce that: 

\begin{lemma}
If $\vc{K}^{(1)} \preceq \vc{K}^{(2)}$ then the corresponding error bound \eqref{eq:COMPbd2} is lower for $\vc{K}^{(1)}$ than $\vc{K}^{(2)}$.
\end{lemma}

Hence, for fixed $K$ the error bound \eqref{eq:COMPbd2} is smallest for the minimal sequence $\vc{K} = (K)$ corresponding to the classical case and largest for the sequence $\vc{K} = (1,1, \ldots, 1)$ where each defective item has its own unique level.

\section{Analysis of tropical DD algorithm: achievability} \label{sec:analysis_DD}

In this section we give an analysis of the performance of tropical DD, which extends that given in \cite{aldridge2014} for the classical $d=1$ case by taking advantage of the information provided by the more varied test outcomes of the tests. 

Our main achievability result ensures success with high probability when we have a number of tests above a certain threshold.

\begin{theorem} \label{thm:DDachievability}
 For $\new > 0$, write
 \[ \pc_r := \pc_r(\new) = \left(1-\frac{\new}{K} \right)^{\sum_{t \leq r} K_t} . \]
 Also define
 \[ \Tinf(\new) := \frac{1}{\new \pc_d} K \ln\frac{N}{K} \]
 and
 \[ T_r(\new) := \frac{1}{{\new\pc_r}} K \ln K_r . \]
 If we take
 \[ T \geq (1+\delta) \max\big\{\Tinf(\new), T_d(\new), T_{d-1}(\new), \dots, T_1(\new) \big\} , \]
 tests, then the error probability of the tropical DD algorithm for a Bernoulli design with $p=\new/K$ tends to $0$.
\end{theorem}

\begin{remark} \label{rem:improve}
Note that in the uniform case where $K_r = K/d$ for all $r$, 
in regimes where $K = N^\alpha$ and $\alpha > 1/2$, the dominant term in Theorem \ref{thm:DDachievability} is 
\[ T_d(\new) = \frac{1}{\new (1-\new/K)^K} \, K \ln\frac{K}{d} ; \]
that is, the maximum over $r$ is achieved at $r=d$, since $\pc_r$ is decreasing in $r$ and $\ln K_r$ is constant.

Further, in this case, since we are able to choose $\new$ to minimise $\Tfin(\new)$, we can maximise $\new(1-\new/K)^K$ by taking $\new=K/(K+1)$ or $p=1/(K+1)$. Note that this does not necessarily mean that we minimise the error probability, since we are minimising the number of tests to control a bound on $\Perr$, not $\Perr$ itself. In other words, asymptotically we require $\ee K \ln(K/d)$ tests.

This means that the tropical performance bound of Theorem \ref{thm:DDachievability} represents an asymptotic reduction of $\ee K \ln d$ tests over the classical bound of $\ee K \ln K$ tests (see \cite[Theorem 1]{aldridge2017}).  While this is a second-order asymptotic term compared with the leading order $\ee K \ln K$ term, it may still represent a valuable improvement in problems of finite size.
\end{remark}

In the next section, we will see a converse result Theorem \ref{thm:DDconverse}, showing that success probability can't be high with a number of tests below a threshold. Further, we show that for certain parameter regimes these two thresholds coincide, showing that we have sharp performance bounds on tropical DD.

\subsection{Proof outline} \label{sec:outline}

To prove Theorem \ref{thm:DDachievability}, we need a general upper bound on the error probability of DD. The key idea is that DD succeeds if and only if each defective item is proven to be such in the second stage of the algorithm.

\begin{definition} \label{def:Lrs}
For any $1 \leq s \leq K_r$, we write $L_{r,s}$ for the number of tests that contain item $i(r,s)$, no other defective item $i(t,u)$ with $t \leq r$, and also no non-defective $\PD(r)$ item. 
\end{definition}

A test that counts towards $L_{r,s}$ is precisely one that discovers $i(r,s)$ to be defective at level $r$. So with this definition, we can say that the tropical DD algorithm succeeds if and only if $L_{r,s} \geq 1$ for all $(r,s)$. Hence we have
\begin{equation} \Perr = \pr \left( \bigcup_{r,s} \{ L_{r,s} = 0 \} \right) \leq
\sum_{r=1}^d \pr \left( \bigcup_{s=1}^{K_r} \{ L_{r,s} = 0 \}  \right),
\label{eq:DDPerrbd2}
\end{equation}

One way we could get $L_{r,s} = 0$ is if there is a large number of potentially intruding non-defectives at this outcome level. We define $G_r = \sum_{j=0}^{r} H_j$ for the number of intruding non-defectives (where $H_j$ is introduced in Definition \ref{def:keydefs}.\ref{it:hdef}) and note that there are $(N-K) - G_r = H_{r+1} + \ldots + H_d + H_\infty$ non-defective items with $\mu_i > r$. We observe that, provided we use sufficiently many tests, $G_r$ is unlikely to be large. Hence, we will condition on $G_r$ being no larger than some threshold level $g_r^*$, to be chosen later. So for each level $r$, the summand in \eqref{eq:DDPerrbd2} can be bound as
\begin{align}
\pr \left( \bigcup_{s=1}^{K_r} \{ L_{r,s} = 0 \} \right)  
&= \pr \left( \left. \bigcup_{s=1}^{K_r} \{ L_{r,s} = 0 \} \,\right|\, G_r \leq g_r^* \right)  \pr(G_r \leq g_r^*) \notag \\
&\qquad {}+ \pr \left( \left. \bigcup_{s=1}^{K_r} \{ L_{r,s} = 0 \} \,\right|\, G_r > g_r^* \right)  \pr(G_r > g_r^*) \notag \\
&\leq \pr \left( \left. \bigcup_{s=1}^{K_r} \{ L_{r,s} = 0 \} \,\right|\, G_r \leq g_r^* \right) + \pr(G_r > g_r^*) \notag \\ 
&\leq K_r \, \pr \left(  L_{r,s} = 0 \mid G_r \leq g_r^* \right) + \pr(G_r > g_r^*)  , \label{eq:firstpart}
\end{align}
where we used the union bound in the last line.

We need to show that both terms in \eqref{eq:firstpart}  are small. The first term being small tells us we're likely to find the level-$r$ defectives provided $G_r$ is not too big; we will show this happens when $T \geq (1+\delta)\max\{T_\infty, T_r \}$. The second term being small tells us that $G_r$ is unlikely to be too big; we will show this happens when $T$ is big; for example, $T \geq (1 + \delta) T_r$ will be plenty.

In Subsection \ref{sec:finding} we will analyse the first term $\pr \left(  L_{r,s} = 0 \mid G_r \leq g_r^* \right)$. In Subsection \ref{sec:intruding} we will bound the second term $\pr(G_r > g_r^*)$. Then in Subsection \ref{sec:complete} we will put the pieces together to prove Theorem \ref{thm:DDachievability}.

\subsection{Finding defectives} \label{sec:finding}

We first describe the joint distribution of certain random variables arising in the analysis of DD. We provide additional notation to that used in Section \ref{sec:notation}.

\begin{definition} \label{def:mrs}
We define the following random variables: 
\begin{enumerate}
\item Write $M_\infty$ for the number of tests which contain no defectives -- and so are negative tests with outcome $\infty$.
\item For $1 \leq r \leq d$, write $M_r$ for the total number of positive tests with outcome $r$.
\item Further, decompose $M_r = \sum_{s = 1}^{K_r} M_{r,s}  + M_{r,+}$ as follows:
\begin{enumerate}
\item For $1 \leq s \leq K_r$, write $M_{r,s}$ for the number of tests that contain  a single item $i(r,s)$ at level $r$ and no other defective item $i(t,u)$ with $t \leq r$;  note that each such test has outcome $r$.
\item Write $M_{r,+}$ for the number of tests which have outcome $r$ but contain multiple defective items at level $r$.
\end{enumerate}
\end{enumerate}
Write $\vc{M}$ for the collection of random variables
$$\vc{M} = \left( M_{1,1}, M_{1,2}, \ldots M_{1,K_1}, M_{1,+}, \ldots,
M_{d,1}, M_{d,2}, \ldots, M_{d,K_d}, M_{d,+}, M_{\infty} \right) $$
(noting this includes the terms in the decompositions of the $M_r$ variables, but not the $M_r$ themselves).
\end{definition}

Note that $M_{r,s} = 0$ means necessarily that $L_{r,s} = 0$, and this is the event we wish to avoid. But first, let us note the joint distribution of $\vc{M}$. 

\begin{lemma} \label{lem:mdist}
The random vector $\vc{M}$ is multinomial with parameters $T$ and $\vc{q}$, where
$$\vc{q} = \left( q_{1,1}, q_{1,2}, \ldots q_{1,K_1}, q_{1,+}, \ldots,
q_{d,1}, q_{d,2}, \ldots, q_{d,K_d}, q_{d,+}, q_{\infty} \right).$$
Here  for each $r$ and $s$:
\begin{align*}
    q_\infty & :=  (1-p)^K , \\
    q_{r,s} &:= \prod_{t < r} (1-p)^{K_t} \left( (1-p)^{K_r-1} p \right), \\
    q_{r} & := \prod_{t < r} (1-p)^{K_t} \left( 1 - (1-p)^{K_r} \right), \\
    q_{r,+} &:= q_r - K_r q_{r,s}. 
\end{align*}
\end{lemma}

\begin{proof}
First, a test is negative if all defective items are absent, with happens with probability $q_\infty =(1-p)^K$. Second, $q_{r,s}$ is the probability all items at levels $t < r$ are absent, that item $i(r,s)$ is present, and that the $K_r - 1$ other items at level $r$ are absent. Third, $q_r$ is the probability of outcome $r$, which happens all items at level $t < r$ are absent, and also it's not the case that all items at level $r$ are absent. Fourth, $q_{r,+}$ is the probability $q_r$ of an outcome $r$ minus the probabilities of a single level-$r$ item being the cause.
\end{proof}

Although the distribution of the crucial variable $L_{r,s}$ seems tricky to derive from first principles, it is much easier once we know $M_{r,s}$ and the number of potentially intruding non-defectives. This is because a test counting towards $M_{r,s}$ will count also towards $L_{r,s}$ provided that no non-defectives intrude on the test too.

\begin{lemma} \label{lem:Lrscond}
The conditional distribution of $L_{r,s}$, given $M_{r,s}$ and $G_r$, is
\begin{equation} \label{eq:lrsdisn}
 L_{r,s} \mid \left\{M_{r,s} = m_{r,s}, G_r =  g_r \right\} \sim \bin(m_{r,s}, (1-p)^{g_r}).\end{equation}
\end{lemma}

\begin{proof}
There are $m_{r,s}$ tests which contain item $i(r,s)$ and no other defective item $i(t,u)$ with $t \leq r$. Each of these $m_{r,s}$ tests independently contributes to $L_{r,s}$ if and only if none of the $g_r$ potentially intruding non-defective items appear in the test. Because of the Bernoulli design, each of those $g_r$ non-defectives appears in the test with probability $p$, so the probability none of them appear is $(1-p)^{g_r}$, independent over the $m_{r,s}$ tests.
\end{proof}

We can now bound the probability that of undesirable event that $L_{r,s} = 0$.

\begin{lemma} \label{lem:PerrDDubfull}
Using a Bernoulli design with parameter $p$, for any $g_r^*$, we have the bound
\begin{equation} \label{eq:PerrDDubfull}
\mathbb P(L_{r,s} = 0 \mid G_r \leq g_r^*) \leq \exp \big( -q_{r,1}(1-pg_r^*)T\big) , 
\end{equation}
where as in Lemma \ref{lem:mdist} we write $q_{r,1} = p (1-p)^{\sum_{t \leq r} K_t - 1}$.
\end{lemma}

\begin{proof}
We start by conditioning on equality $G_r = g_r$.
Noting that by Lemma \ref{lem:mdist} $M_{r,s} \sim \bin(T, q_{r,1})$ and that $ \pr( \bin(m,q) = 0) = (1-q)^m$,
 we can write
\begin{align}
\pr( L_{r,s} = 0 \mid G_r = g_r)
&= \sum_{m=0}^T \pr(M_{r,s} = m) \, \pr( L_{r,s} = 0 \mid G_r = g_r, M_{r,s} = m) \nonumber \\
&=  \sum_{m=0}^T \binom{T}{m} q_{r,1}^m (1-q_{r,1})^{T-m} \,  (1 - (1-p)^{g_r})^m \nonumber \\
& = (1- q_{r,1} (1-p)^{g_r} )^T \nonumber \\
&\leq \exp \left( - q_{r,1} (1-p)^{g_r} T \right) \nonumber \\
&\leq \exp \left( -q_{r,1} (1-p g_r) T \right) \label{eq:expand2} .
\end{align}
From the second to the third line, we used the binomial theorem, and then we used Bernoulli's inequality in the form $(1-p)^g \geq 1-p g$.

Note that \eqref{eq:expand2} is increasing in $g_r$. Thus we can bound \eqref{eq:PerrDDubfull} by the worst-case conditioning, where $G_r = g_r^*$.
\end{proof}

\subsection{Intruding non-defectives} \label{sec:intruding}

Recall  that $G_r$ is the number of non-defectives that could intrude into tests with outcome $r$. Our goal is to bound that tail probability of $G_r$ in \eqref{eq:firstpart}.

\begin{lemma} \label{lem:gdist2}
Write
\[ \ol{M}_r := M_\infty + M_d + M_{d-1} + \cdots + M_{r+1} \]
for the number of tests with outcomes higher than $r$. Then $\ol{M}_r$ has distribution
\begin{equation} \label{eq:msumdist}
\ol{M}_r \sim \bin \left( T, \;\; \prod_{t \leq r} (1-p)^{K_t} \right)
\end{equation}

Further, the conditional distribution of $G_r$ given $\vc{M}$ is
\begin{equation} \label{eq:gconddist}
 G_r \mid \{\vc{M} = \vc{m}\} \sim \bin( N-K, (1-p)^{m_r^*}) , \end{equation}
where $m_r^* = m_{r+1} + \cdots + m_d + m_\infty$.
\end{lemma}

\begin{proof}
By standard properties of the multinomial (see \cite[Lemma 30]{aldridge2014}),
\[ \ol{M}_r \sim \bin \left( T, q_{r+1} + \ldots + q_d + q_\infty \right) . \]
But
\begin{equation} \label{eq:qsum}
 q_{r+1} + \ldots + q_d + q_\infty = \prod_{t \leq r} (1-p)^{K_t}, 
\end{equation}
since it forms a collapsing sum. This proves the first distribution.

Given $\ol{M}_r = m_r^*$, each of the $N - K$ non-defectives will be independently counted in $G_r$ provided they don't appear in any of the $m_r^*$ tests with outcomes higher than $r$. By the Bernoulli design structure, each item is independently counted with probability $(1-p)^{m_r^*}$. This proves the second distribution.
\end{proof}

We can calculate the expectation of $G_r$ by conditioning on $\ol{M}_r$.

\begin{lemma} \label{lem:egr}
$\ep G_r \leq (N-K) \exp(- p \pc_r T)$.
\end{lemma}

\begin{proof}%[Proof of Lemma \ref{lem:egr}]
We use the facts that \eqref{eq:msumdist} gives that $\ol{M}_r \sim \bin(T,\pc_r)$ and
Lemma \ref{lem:gdist2} gives that $G_r \mid \{\ol{M}_r = m^*_r\} \sim \bin \left( N-K, (1-p)^{m^*_r} \right)$. Hence we can use the binomial theorem to write
\begin{align*}
\ep G_r &= \sum_{m=0}^T \pr(\ol{M}_r = m)\, \ep[ G_r \mid \ol{M}_r = m] \\
&= \sum_{m=0}^T \binom{T}{m} \pc_r^m (1-\pc_r)^{T-m} \, (N-K) (1-p)^m \\
&= (N-K) (1- \pc_r + \pc_r(1-p))^T,
\end{align*}
and the result follows.
\end{proof}

We will choose the threshold $g_r^*$ to be just slightly bigger than this expectation; specifically, we take $g_r^* = (N-K) \exp(-p \pc_r T(1-\epsilon))$ for some $\epsilon >0$ to be determined later.

\begin{lemma} \label{lem:egr2}
With $g_r^* = (N-K) \exp(-p \pc_r T(1-\epsilon))$, we have
\[ \mathbb P(G_r > g_r^*) \leq \exp(-p\pc_r T\epsilon) . \] 
\end{lemma}

\begin{proof}
This is a simple application of Markov's inequality. Using Lemma \ref{lem:egr}, we get
\begin{align*}
\pr(G_r > g_r^*) &\leq \frac{\ep G_r}{g_r^*} \\
&\leq \frac{(N-K) \exp(-p \pc_r T)}{(N-K) \exp(-p T \pc_r(1-\epsilon))} \\
& = \exp(-p \pc_r T \epsilon) . \qedhere
\end{align*}
\end{proof}

\subsection{Completing the proof} \label{sec:complete}

We are now ready to complete the proof of our main result.

\begin{proof}[Proof of Theorem \ref{thm:DDachievability}]
From \eqref{eq:DDPerrbd2} and \eqref{eq:firstpart}, we had got as far as the bound
\begin{equation} \label{eq:DDbound2}
\Perr \leq \sum_{r=1}^d \big[ K_r \, \pr \left(  L_{r,s} = 0 \mid G_r \leq g_r^* \right) + \pr(G_r > g_r^*) \big] ,
\end{equation}
and in Subsection \ref{sec:intruding} we had chosen $g_r^* = (N-K) \exp(-p \pc_r T(1-\epsilon))$, with $\epsilon$ still to be fixed.
We need to show that, in each summand of \eqref{eq:DDbound2}, both the first and second terms tend to $0$.

We begin with the first term. From Lemma \ref{lem:PerrDDubfull},
we have the bound
\begin{align*} 
K_r \, \mathbb P(L_{r,s} = 0 \mid G_r \leq g_r^*)
  &\leq K_r \exp \big( -q_{r,1}(1-pg_r^*)T\big) \\
  &= \exp \left( \ln K_r - T \pc_r p \frac{1-p g_r^*}{1-p} \right)  ,
\end{align*}
where we have used that $q_{r,1} = \pc_r p/(1-p)$. The condition $T \geq (1+\delta) T_r$ means that $T \pc_r p \geq (1 + \delta) \ln K_r$, so we get
\begin{equation*} 
K_r \, \mathbb P(L_{r,s} = 0 \mid G_r \leq g_r^*)
  \leq \exp \left( \ln K_r \left( 1 - (1+\delta) \frac{1-p g_r^*}{1-p} \right) \right) .
\end{equation*}
This tends to 0 so long as $p g_r^*$ tends to $0$, which we will now check.

Since $T \geq (1+\delta) \Tinf$ and $\pc_r \geq \pc_d$, we know that $T p \pc_r \geq (1+\delta) \ln (N/K) $. With  $g_r^* = (N-K) \exp(-p \pc_r T(1-\epsilon))$, we therefore have
$$p g_r^* \leq \left( \frac{N}{K} \right) \exp \left( - T p \pc_r (1-\epsilon) \right) \leq  \left( \frac{N}{K} \right)^{1-(1-\epsilon)(1+\delta)}.$$ 
This means that $p g_r^* \rightarrow 0$ is indeed guaranteed by choosing $\epsilon < \delta/(1+\delta)$.

Now the second term. From Lemma \ref{lem:egr2}, we have
\[ \mathbb P(G_r > g_r^*) \leq \exp(-p\pc_r T\epsilon) < \exp(-p\pc_r T\delta / (1+\delta)) , \]
since we have just chosen $\epsilon < \delta/(1+\delta)$.
The condition $T \geq (1+\delta) T_r$ gives us  $p \pc_r T/(1+\delta)  = \ln K_r$, so this term does indeed tend to $0$.

Since we have shown that all the terms in \eqref{eq:DDbound2} tend to zero, the proof is complete.
\end{proof}

\section{Converse results} \label{sec:DDconverse}

Our achievability result Theorem \ref{thm:DDachievability} shows that tropical DD can succeed with 
\begin{equation} \label{eq:totaltests}
T \geq (1+\delta) \max\{\Tinf(\new), T_d(\new), T_{d-1}(\new), \dots, T_1(\new) \} \qquad
\text{tests.} \end{equation}
We can use similar ideas to provide a converse result for the tropical DD algorithm.

\begin{theorem} \label{thm:DDconverse}
For a given $\new > 0$, in the limiting regime  where  
\[ T \leq (1-\delta) \max\{T_d(\new), T_{d-1}(\new), \dots, T_1(\new) \} \] 
then the error probability of the tropical DD algorithm for a Bernoulli design with $p=\new/K$ tends to $1$.
\end{theorem}

Note that the difference between the achievablility and the converse results is the lack of the $\Tinf$ term in the converse.

We will prove Theorem \ref{thm:DDconverse} for tropical DD in Subsection \ref{sec:DDconpf}. In addition, we will show in Subsection \ref{sec:gencon} that this same bound acts as a more general `algorithm-independent' converse for Bernoulli designs.

\subsection{Proof for tropical DD} \label{sec:DDconpf}

The key to proving Theorem \ref{thm:DDconverse} is to observe that tropical DD will definitely fail if any $M_{r,s} = 0$, since that means that item $i(r,s)$ never appears without at least one other defective item with which it could be confused, so $L_{r,s}$ is certainly $0$ too.

Thus we start with the bound
\begin{equation}
 \Perr  \geq  \pr \left( \bigcup_{r,s} \{ M_{r,s} = 0 \} \right) .
\end{equation}
By picking out just the defective items at a given level $r = r^*$, we have
\begin{equation} \label{eq:conversefirst}
\Perr  \geq  \pr \left( \bigcup_{s=1}^{K_{r^*}} \{ M_{r^*,s} = 0 \} \right)
\end{equation}

As in \cite[Eq.~(13)]{aldridge2014}, we define the function
\begin{equation} \label{eq:phidef}
\phi_K(q, T) :=
\sum_{j=0}^{K} (-1)^{j} \binom{K}{j} (1 - j q)^T.
\end{equation}
We will bound the error probability in terms of $\phi_K$ as follows.

\begin{lemma}
For $1 \leq r^* \leq d$, the error probability of tropical DD is bounded below by
\begin{equation} 
 \Perr  \geq   1- \phi_{K_r^*}( q_{r^*,1}, T), \label{eq:fullbdphi} 
\end{equation}
\end{lemma}

\begin{proof}
We follow the general idea from \cite{aldridge2014}.

We can calculate \eqref{eq:conversefirst} using the inclusion--exclusion formula
\begin{equation} \label{eq:incexc}
\pr \left( \bigcup_{s=1}^{K_{r^*}} \{ M_{r^*,s} = 0 \} \right) 
= \sum_{|S| = 0}^{K_{r^*}} (-1)^{|S|} \,  \pr \left( \bigcap_{s \in S} \{ M_{r^*,s} = 0 \} \right) ,
\end{equation}
where the sum is over subsets $S$ of $\mathcal K_r$.
By the multinomial distribution form of Lemma \ref{lem:mdist}, we have
\begin{align}
\pr \left( \bigcap_{s \in S} \{ M_{r^*,s} = 0 \} \right) &= \binom{T}{0, 0, \ldots, 0, T} \left( \prod_{s \in S} q_{r^*,s}^0 \right) \left( 1 - \sum_{s \in S} q_{r^*,s} \right)^T \nonumber \\
&= \left( 1 - \sum_{s \in S} q_{r^*,s} \right)^T \notag \\
&= \left(1 - |S| q_{r^*,1}\right)^T .
\label{eq:MNintersect} 
\end{align}
Substituting \eqref{eq:MNintersect} into \eqref{eq:incexc} gives
\[ \pr \left( \bigcup_{s=1}^{K_{r^*}} \{ M_{r^*,s} = 0 \} \right) 
= \sum_{|S| = 0}^{K_{r^*}} (-1)^{|S|} \,  \left(1 - |S| q_{r^*,1}\right)^T . \]
Collecting together the summands according to the value of $|S| = j$ gives the result.
\end{proof}

We bound this quantity from below by deducing an upper bound on $\phi_K$. 

\begin{lemma} \label{lem:phiupperbd}
For all $K$, $q$ and $T$ we can bound
\begin{equation} \phi_K(q,T) \leq  \exp \left( - \frac{ K (1-q)^{T+1}}{1 + K q(1-q)^T} \right) \label{eq:newphiUB2}
\end{equation}
\end{lemma}

\begin{proof}
See Appendix \ref{sec:DDlbpf}.
\end{proof}

We can now finish our proof of the converse for tropical DD.

\begin{proof}[Proof of Theorem \ref{thm:DDconverse}]
By hypothesis, $T \leq (1-\delta) \max_r T_r$. So pick some level $r^*$ such that $T \leq (1-\delta) T_{r^*}$.

We had already reached the bound \eqref{eq:fullbdphi}:
\begin{equation*} 
 \Perr  \geq   1- \phi_{K_r^*}( q_{r^*,1}, T) .
\end{equation*}
We now combine this with Lemma \ref{lem:phiupperbd}. We deduce that
\begin{equation} \label{eq:stepA0}
  \Perr \geq 1- \exp \left( - \frac{ K_{r^*} (1-q_{r^*,1})^{T+1}}{1 + K_{r^*} q_{r^*,1} (1-q_{r^*,1})^T} \right) .
\end{equation}
The exponential term here is of the form $\exp(-(1-q)u/(1+qu))$, with $u = K_{r^*}(1-q_{r^*,1})^T$. Since $\exp(-(1-q)u/(1+qu))$ increases as $u$ decreases (for fixed $q$), it suffices to bound $K_{r^*}(1-q_{r^*,1})^T$ from below, which we do now.

Since $q_{r^*,1} = p \pc_{r^*}/(1-p)$ we know that 
\[ \frac{q_{r^*,1}}{1-q_{r^*,1}} = \frac{p \pc_{r^*}}{1-p(1+\pc_{r^*})} . \]
Combining this with $T \leq (1-\delta)T_{r^*} = (1-\delta) K  \ln K_{r^*}/(\new \pc_{r^*})$ gives
\begin{align} 
\frac{T q_{r^*,1}}{1-q_{r^*,1}}
& \leq \frac{(1-\delta) P}{\new \pc_{r^*}} \frac{p \pc_{r^*}}{1-p(1+\pc_{r^*})} \ln K_{r^*} \notag \\
&= \frac{(1-\delta)}{1-p(1+\pc_r^*)} \ln K_r^* \notag \\
&\leq (1-c) \ln K_r^*, \label{eq:stepAhalf}
\end{align}
for some $c > 0$, for $K_{r^*}$ sufficiently large. This gives us the lower bound
\begin{align} 
K_{r^*} (1-q_{r^*,1})^T
&= K_{r^*} \exp( T \log(1-q_{r^*,1})) \notag \\
&\geq K_{r^*} \exp \left( - \frac{T q_{r^*,1}}{1-q_{r^*,1}} \right) \notag \\
&\geq K_{r^*}^c,\label{eq:Krbd}
\end{align}
where we used $\log(1-q) \geq -q/(1+q)$ for $q > -1$ and \eqref{eq:stepAhalf}.

Using the bound \eqref{eq:Krbd} in \eqref{eq:stepA0}, we get
\begin{align} 
  \Perr &\geq 1- \exp \left( - \frac{ K_{r^*} (1-q_{r^*,1})^{T+1}}{1 + q_{r^*,1} K (1-q_{r^*,1})^T} \right) \nonumber \\
        &\geq 1- \exp \left( - \frac{ K_{r^*}^c (1-q_{r^*,1})}{1 + q_{r^*,1} K_{r^*}^c} \right)  \notag \\ %\label{eq:stepA1}
        &= 1 - \exp \left( - \frac{ (1- \pc_{r^*}/(K-1)) K_{r^*}^c}{1 + \pc_{r^*} K_{r^*}^c/(K-1)} \right). \label{eq:stepA2}
\end{align}
where \eqref{eq:stepA2} follows since 
$q_{r^*,1} = p \pc_{r^*}/(1-p) = \pc_{r^*}/(K-1)$.

Finally, that bound \eqref{eq:stepA2} is asymptotically equivalent to $1 - \exp(-K_{r^*}^c)$, which tends to 1 as $K \rightarrow \infty$. This completes the proof.
\end{proof}

\subsection{Algorithm-independent converse} \label{sec:gencon}

In fact, our DD-specific converse, Theorem \ref{thm:DDconverse}, helps give a converse bound for \emph{all} algorithms with a Bernoulli design.

We can write $\PERR{optimal}{T}$ for the minimum error probability that can be achieved by any algorithm. The key observation is that in Theorem \ref{thm:DDconverse} we find that with too few tests there is a good chance that some item $i(r^*,s)$ appears in no tests without other items of the same level, so a satisfying vector can be formed without it.

\begin{theorem} \label{thm:genconv}
For a given $\new > 0$ and any $T \leq (1-\delta) \Tfin(\new)$, the error probability of the optimal algorithm $\PERR{optimal}{T}$ for a Bernoulli design with parameter $\new/K$ is bounded away from zero, even for algorithms which are given access to the values of $K_i$.
\end{theorem}

\begin{proof}
First note that for any $T' \geq 1$ the $\PERR{optimal}{T} \geq \PERR{optimal}{T+T'}$, since there exists a (possibly suboptimal) algorithm using $T+T'$ tests which simply ignores the last $T'$ tests and applies the optimal $T$-test algorithm to the remaining tests. Hence it will be sufficient to bound $\PERR{optimal}{T}$ away from zero for $T = (1-\delta) \max_r T_r$, as the same bound will hold for all $T \leq (1-\delta) \max_r T_r$.

We argue as in \cite{aldridge2017}. Recall that $M_{r,s}$ is the number of tests that have a chance of proving $i(r,s)$ is defective at level $r$, and $H_r$ is the number of non-defective items in $\PD(r)$. The key idea is this: Suppose for some $r^*$ that both $A  = \bigcup_{s=1}^{K_{r^*}} \{ M_{{r^*},s} = 0 \}$ and $B = \{ H_{r^*} \geq 1 \}$ occur. The event $A$ means that there is some item $i(r^*,t)$ which never appears without some other item $j$ of level $U_j \leq r$, and the event $B$ means that there is some non-defective item which is a possible defective at that level. So we could form an alternative satisfying vector from the true vector $\vc{U}$ by swapping the entries in $\vc{U}$ of these two items. Hence, if $A \cap B$ occurs, then there are at least two satisfying vectors with the correct number of items at each level, so we the success probability can only be at most $1/2$.

The probability of an intersection can always be bounded with $\pr(A \cap B) \geq \pr(A) - \pr(B^c)$, so the error probability for any algorithm satisfies
\begin{eqnarray*}
 \Perr \geq \frac{1}{2} \, \pr \left( 
 \bigcup_{s=1}^{K_{r^*}} \{ M_{r^*,s} =0 \} \right)
 -  \frac{1}{2}\, \pr(H_{r^*} =  0).
 \end{eqnarray*}
Now the first term involves exactly the term we have controlled in Theorem \ref{thm:DDconverse}, so we know it is larger than $1/4$ in the regime of interest for $K$ sufficiently large. Hence, to bound the error probability away from zero it will be sufficient to prove that $\pr(H_{r^*} = 0) \leq 1/4$.

We will prove this in a series of technical lemmas in Appendix \ref{sec:algindlem}:
\begin{enumerate}
\item In Lemma \ref{lem:new1}, we will show that
\[ \pr(H_{r^*} = 0) \leq \pr(G_{r^*} =0) + \ep (1-p)^{M_{r^*}} . \]
We deal with the two terms separately.
\item In Lemma \ref{lem:gbd}, we will show that the first term is bounded by
\begin{equation} \pr(G_{r^*} = 0) \leq \left( 1- (1-p)^{m_{r^*}^*} \right)^{N-K} + \exp \left( -\frac{ \delta^2 T \pc_{r^*}}{2} \right) , \label{eq:glprb}
\end{equation}
where $m_\ell^* = T \pc_\ell (1+\delta)$.
\item In Lemma \ref{lem:new2}, we show that the second term is bounded by 
\begin{equation}
\ep (1-p)^{M_{r^*}} \leq \exp( - p \pc_{r^*} d_{r^*} T) , \label{eq:mlglbd}
\end{equation}
where $d_\ell = (1-p)^{-K_\ell} - 1$.
\end{enumerate}

Recall that we consider $T = (1-\delta) T_{r^*} = (1-\delta) K \ln K_{r^*}/(\new \pc_{r^*})$, for the maximising $r^*$. Since $p \pc_{r^*} T = (1-\delta) \ln K_r^*$, we know that $(1-p)^{m_{r^*}} \aequi K^{-(1-\delta^2)}$, so that both terms in \eqref{eq:glprb} tend to zero. Similarly, \eqref{eq:mlglbd} also tends to zero for this choice of $\ell$ and $T$. This completes the proof.
\end{proof}

\section{Discussion}
In this paper, we have considered the tropical group testing model of Wang {\em et al.} \cite{wang} in a small-error setting. We have described small-error algorithms in Section \ref{sec:algo}. We demonstrated the empirical performance of these algorithms in Section \ref{sec:simulation}, showing that tropical DD and tropical SCOMP outperform their classical counterparts. We performed theoretical analysis of the tropical COMP algorithm in Section \ref{sec:analysis_COMP} and of the DD algorithm in Sections \ref{sec:analysis_DD} and \ref{sec:DDconverse}, proving that in certain parameter regimes the tropical DD algorithm is asymptotically optimal.

We briefly mention some open problems. Further work could explore test designs with near-constant column weights in the tropical setting, as these designs show a gain in performance in the classical case (see \cite{aldridge2019}), and Figure \ref{fig:diffdesigns} suggests the same may well be true here. The results could be made more practically valuable by developing bounds in a noisy setting, under a variety of noise models similar to those described in \cite[Chapter 4]{aldridge2019}. Also, there is potential to extend the results in this paper by considering models with random defectivity levels, as illustrated in Figure \ref{fig:diffd}. It may also be mathematically interesting to develop small-error algorithms and bounds using the delay matrix approach of \cite{wang}.

\section*{Acknowledgements}
This work was carried out while Vivekanand Paligadu was on a University of Bristol School of Mathematics undergraduate summer bursary placement, funded by Mark Williams Alumni Funds.

%\bibliographystyle{abbrv}
%\bibliography{papers}

\appendix

\section{Proof of the counting bound} \label{sec:counting}

We now prove Theorem \ref{thm:counting}, using an argument closely adapted from \cite[Proof of Theorem 3.1]{baldassini2013}.

\begin{proof}[Proof of Theorem \ref{thm:counting}]
Given $N$, $\Dset$ and $\vc{K} = (K_1, \ldots, K_d, K_{\infty})$, we write 
$$ \Sigma_{N,\vc{K}} = 
\left\{ \vc{u} \in \Dset^N: | \{ j: u_j = r \}| = K_r \text{ for all } r \in \Dset \right\}
$$ 
for the collection of vectors with the correct number of each type of component. Note that $| \Sigma_{N,\vc{K}}| = \binom{N}{\vc{K}}$ is precisely the multinomial term arising in Theorem \ref{thm:counting}. Further, we write $\defvec$ for the true defectivity vector.

The testing procedure naturally 
defines a mapping $\theta: \Sigma_{N,\vc{K}} \rightarrow \Dset^T$. That is, given a putative defectivity vector $\vc{u} \in \Sigma_{N, \vc{K}}$, write $\theta(\vc{u})$ to be the vector of test outcomes. For each vector $\vc{y} \in \Dset^T$, write $\AAA{y} \subseteq
\Sigma_{N, \vc{K}}$ for the inverse image of $\vc{y}$ under $\theta$,
$$ \AAA{y} = \theta^{-1}(\vc{y}) = 
\left\{ \vc{u} \in \Sigma_{N,\vc{K}}: \theta( \vc{u}) = \vc{y} \right \}.$$
As described in \cite{baldassini2013}, a decoding algorithm aims to mimic the inverse image map $\theta^{-1}$. Exactly as in \cite{baldassini2013}, if size $\AAS{y} \geq 1$ we cannot do better than to pick uniformly among $\AAA{y}$, with success probability $1/\AAS{y}$. (We can ignore $\vc{y}$ with size $\AAS{y} = 0$ because these are outcomes which do not occur for any possible defectivity vector).

We can find the success probability by conditioning over all the equiprobable values of the defective set:
\begin{align*}
  \Psuc 
        &= \sum_{u \in \Sigma_{N, \vc{K}}} \frac{1}{\binom{N}{\vc{K}}}  \pr \left( \suc \mid  \defvec= \vc{u} \right)   \\
    &= \frac{1}{\binom{N}{\vc{K}}} \sum_{\vc{u} \in \Sigma_{N,\vc{K}}} \left( \sum_{\vc{y} \in \Dset^T} \II( \theta(\vc{u}) = \vc{y}) \right) \pr \left( \suc \mid  \defvec= \vc{u} \right)  \\
    & \leq \frac{1}{\binom{N}{\vc{K}}} \sum_{ \vc{u} \in \Sigma_{N,\vc{K}}}  \sum_{\vc{y} \in \Dset^T: \AAS{y} \geq 1} 
         \II( \theta(\vc{u}) = \vc{y}) \frac{1}{\AAS{y}} \\
    &= \frac{1}{\binom{N}{\vc{K}}} \sum_{\vc{y} \in \Dset^T: \AAS{y} \geq 1}   \frac{1}{\AAS{y}}
         \left( \sum_{\vc{u} \in \Sigma_{N,\vc{K}}} \II( \theta(\vc{u}) = \vc{y}) \right) \\
    &= \frac{1}{\binom{N}{\vc{K}}} \sum_{\vc{y} \in \Dset^T: \AAS{y} \geq 1}
         \frac{1}{\AAS{y}}  \AAS{y}  \\
    &= \frac{| \{ \vc{y} \in \Dset^T: \AAS{y} \geq 1 \} |}{\binom{N}{\vc{K}}} \\
    &\leq \frac{(d+1)^T}{\binom{N}{\vc{K}}},
\end{align*}
since $\Dset^T$ is a set of size $(d+1)^T$.
\end{proof}

\section{Lemmas for converse results} \label{sec:conv}

% \subsection{Properties of $\phi_K$} \label{sec:DDlbpf} Previous version (gave warnings)
\subsection{Properties of \texorpdfstring{$\phi_K$}{phiK}} \label{sec:DDlbpf}

We seek to prove Lemma \ref{lem:phiupperbd}, which gives an upper bound on $\phi_K$ and hence a lower bound on $\Perr$. We first need two technical results.

First, recall from \cite[Lemma 32]{aldridge2014} that $\phi_K(q,T)$ is increasing in $q$ for fixed $K$ and $T$. We can provide a similar analysis to show that:

 \begin{lemma} \label{lem:phidecK}
For fixed $T$ and $q$, the function $\phi_K(q,T)$ is decreasing in $K$.
 \end{lemma}
 \begin{proof} We use Pascal's identity $\binom{K+1}{j} = \binom{K}{j} + \binom{K}{j-1}$ to expand
\begin{eqnarray}
\phi_{K+1}(q,T) - \phi_K(q,T) 
& = & \sum_{j=0}^{K+1} (-1)^j \left( \binom{K}{j} + \binom{K}{j-1}  -\binom{K}{j} \right) (1-j q)^T  \nonumber \\
& = & \sum_{j=1}^{K+1} (-1)^j \binom{K}{j-1} (1-j q)^T \label{eq:firstder} \\
& = & \sum_{\ell=0}^{K} (-1)^j \binom{K}{\ell} (1- (\ell+1) q)^T \nonumber \\
& = & - (1-q)^T \sum_{\ell=0}^{K} (-1)^\ell \binom{K}{\ell}  \left(1- \ell \frac{q}{1-q} \right)^T \nonumber \\
& = & - (1-q)^T \phi_{K}(q/(1-q),T) \leq 0, \nonumber
\end{eqnarray}
using the fact (as in \cite{aldridge2014})
that $(1- (\ell+1) q) = (1 -q)( 1- \ell q/(1-q))$.
 \end{proof}

Note that an alternate proof of this is to compare \eqref{eq:firstder} with terms in the proof of \cite[Lemma 32]{aldridge2014} to deduce that
$$ \phi_{K+1}(q,T) - \phi_K(q,T) = - \frac{1}{(K+1)(T+1)}
\frac{\partial}{\partial q} \phi_{K+1}(q,T+1) \leq 0.$$

The second technical result is the following:

 \begin{lemma} \label{lem:phidecT} The function $\phi_K(q, T)$ satisfies
\begin{equation} \label{eq:phimultbd}
\phi_K(q, T+1) \geq \phi_K(q,T) \left( 1+ K q (1-q)^T \right).
\end{equation}
\end{lemma}
\begin{proof}
As in \cite[Proof of Lemma 32]{aldridge2014},
since $j \binom{K}{j} = K \binom{K-1}{j-1}$ we know that
\begin{eqnarray*}
\lefteqn{\phi_K(q, T+1) - \phi_K(q,T) } \\
& = &
- \sum_{j=0}^{K} (-1)^{j} \binom{K}{j} (1 - j q)^{T} j q \\
& = & - K q
\sum_{j=1}^{K} (-1)^{j} \binom{K-1}{j-1} (1 - j q)^T \\
& = & K q (1-q)^T
\sum_{j=1}^{K} (-1)^{j-1} \binom{K-1}{j-1} \left(1 - (j-1) \frac{q}{1-q} \right)^T \\
& = & K q (1-q)^T  \phi_{K-1} \left( \frac{q}{1-q}, T \right),  \end{eqnarray*}
since $1- j q = (1-q) (1 - (j-1) q/(1-q))$. Now, using the fact that $\phi_K(q,T)$ is increasing in $q$ and decreasing in $K$, we know that
$$ \phi_{K-1} \left( \frac{q}{1-q}, T \right)
\geq \phi_K \left( \frac{q}{1-q}, T \right) \geq \phi_K(q,T),$$
and the result follows.
\end{proof}

We can now prove the bound on $\phi_K$.

\begin{proof}[Proof of Lemma \ref{lem:phiupperbd}]
Using Lemma \ref{lem:phidecT} on the terms of a collapsing product and the fact that $\phi_K(q,S) \leq 1$ for all $S$, we can consider $S \rightarrow \infty$ to write
\begin{align}
\phi_K(q,T) & = \phi_K(q,S+1) \prod_{t=T}^S \frac{\phi_K(q,t)}{\phi_K(q,t+1)} \nonumber \\
&\leq  \phi_K(q,S+1) \exp \left( - \sum_{\ell=T}^S
 \log(1 + K q(1-q)^\ell) \right) \nonumber \\
 &\leq \exp \left( - \sum_{\ell=T}^\infty \frac{K q(1-q)^\ell}{1 + K q(1-q)^\ell} \right), \label{eq:step1} \\
 &\leq \exp \left( - \int_{T}^\infty \frac{K q(1-q)^\ell}{1 + K q(1-q)^\ell} d\ell \right), \label{eq:step2} \\
 & = \exp \left( - \frac{\log(1 + K q(1-q)^T)}{-\log(1-q)}
\right) \label{eq:newphiUB} 
 \end{align}
 where a) \eqref{eq:step1} follows using the fact that $-\log(1+x) \leq -x/(1+x)$, b) \eqref{eq:step2} follows since
$f(x) := K q(1-q)^x/(1+ K q (1-q)^x)$ is decreasing in $x$ so we can bound $\sum_{\ell=T}^\infty f(\ell) 
 \geq  \int_{T}^\infty f(\ell) d\ell$, c) 
 \eqref{eq:newphiUB2}
 follows using the fact that (in the numerator) $\log(1+x) \geq x/(1+x)$
 and (in the denominator) $-\log(1-q) \leq q/(1-q)$.
 \end{proof}

\subsection{Lemmas for the algorithm-independent converse} \label{sec:algindlem}

This section consists of various lemmas used in the proof of Theorem \ref{thm:genconv}.

\begin{lemma} \label{lem:hdist}
For any $\ell$, the conditional distribution of $H_\ell$ is given by
$$H_\ell \mid \left\{ M_\ell = m_\ell, G_\ell = g_\ell \right\} 
\sim \bin(g_\ell, 1-(1-p)^{m_\ell}). $$
\end{lemma}

\begin{proof}
Suppose one has already examined the tests with outcomes $\infty, d, \dots, \ell+1$. There remain $G_\ell = g_\ell$ non-defective items which have not appeared in any of these tests. They will each independently contribute to $H_\ell$ unless they avoid appearing in each of the $m_\ell$ tests with outcome $\ell$. Because of the Bernoulli design of the matrix, this will happen with probability $1-(1-p)^{m_\ell}$. 
\end{proof}

\begin{lemma} \label{lem:new1}
The probability we have no non-defectives in $\PD(\ell)$ is bounded above by
\[ \pr(H_{r^*} = 0) \leq \pr(G_{r^*} =0) + \ep (1-p)^{M_{r^*}} . \]
\end{lemma}

\begin{proof}
We know from Lemma \ref{lem:hdist} that
$H_\ell \mid \{M_\ell=m_\ell, G_\ell= g_\ell\} \sim \bin(g_\ell, 1- (1-p)^{m_\ell})$. Since $\pr( \bin(g,1-Q) = 0) = (1-Q)^g$ we know that for any $g_\ell$
\begin{align}
\pr(H_\ell = 0) &= \sum_{m_\ell, g_\ell} \pr(H_\ell = 0 \mid M_\ell=m_\ell, G_\ell= g_\ell) \, \pr(M_\ell=m_\ell, G_\ell= g_\ell) \nonumber \\
&= \sum_{m_\ell, g_\ell} (1-p)^{m_\ell g_\ell} \, \pr(M_\ell = \ell, G_\ell = g_\ell) \notag \\
&= \ep (1-p)^{M_\ell G_\ell} \label{eq:equiv} \\
&\leq \pr(G_\ell =0) + \ep (1-p)^{M_\ell} . \label{eq:hbd} 
\end{align}
Here, we got \eqref{eq:hbd} by considering separately the case where $G_\ell =0$, and by noting that if $G_\ell \geq 1$ then
$(1-p)^{M_\ell G_\ell} \leq (1-p)^{M_\ell}$.
\end{proof}

\begin{lemma} \label{lem:gbd} Writing $m_\ell^* = T \pc_\ell (1+\delta)$ we can bound
\begin{equation} \pr(G_\ell = 0) \leq \left( 1- (1-p)^{m_\ell^*} \right)^{N-K} + \exp \left( -\frac{ \delta^2 T \pc_\ell}{2} \right). \notag
\end{equation}
\end{lemma}

\begin{proof}
%We can provide a bound similar to \eqref{eq:condbd2}, since
Using Lemma \ref{lem:gdist2} we know $\pr(G_\ell =0 \mid \ol{M}_\ell = m) = \left( 1- (1-p)^m \right)^{N-K}$, which is increasing in $m$. Therefore we can write
\begin{align}
 \pr(G_\ell =0) 
&= \pr( G_\ell =0 \mid \ol{M}_\ell \leq m_\ell^*) \, \pr(\ol{M}_\ell = m_\ell^*) \nonumber \\ 
& \qquad {}+ \pr( G_\ell =0 \mid \ol{M}_\ell > m_\ell^*) \, \pr(\ol{M}_\ell > m_\ell^*)  \nonumber \\
&\leq \pr( G_\ell =0 \mid \ol{M}_\ell \leq m_\ell^*) +\pr(\ol{M}_\ell > m_\ell^*)  \nonumber \\
&\leq \pr( G_\ell =0 \mid \ol{M}_\ell = m_\ell^*) +\pr(\ol{M}_\ell > m_\ell^*)  \nonumber \\
 &\leq \left( 1- (1-p)^{m_\ell^*} \right)^{N-K}
 + \pr(\ol{M}_\ell > m_\ell^*). 
\end{align}
Recall that $\ol{M}_\ell \sim \bin(T, \pc_\ell)$, so $m_{\ell}^* = T \pc_\ell(1+\delta)$ is slightly bigger than the expectation, and we can bound the second term by Chernoff's inequality.
\end{proof}

\begin{lemma} \label{lem:new2}
We have the bound
\begin{equation}
\ep (1-p)^{M_{r^*}} \leq \exp( - p \pc_{r^*} d_{r^*} T) , \notag
\end{equation}
where $d_\ell = (1-p)^{-K_\ell} - 1$.
\end{lemma}

\begin{proof}
From Lemma \ref{lem:mdist} we know that $M_\ell \sim \bin(T, \pc_\ell d_\ell)$. Hence, using the binomial theorem:
\begin{align}
 \ep (1-p)^{M_\ell}  &= \sum_{m=0}^T \binom{T}{m} (\pc_\ell d_\ell)^m (1- \pc_\ell d_\ell)^{T-m} (1-p)^{m} \nonumber \\
&= \left(1- p \pc_\ell d_\ell \right)^T \notag \\
&\leq \exp( - p \pc_\ell d_\ell T) . \notag \qedhere
\end{align}
\end{proof}

\end{document}